\documentclass[12pt]{article}

\usepackage{mathptmx}      % use Times fonts if available on your TeX system
\usepackage[normalem]{ulem}

\usepackage{times}
\usepackage{soul}
\usepackage{url}
\usepackage[utf8]{inputenc}
\usepackage{amsmath}
\usepackage{booktabs}
\usepackage{multirow}
\usepackage{tabularx}

\usepackage{amsthm}
\usepackage{amssymb}
\usepackage{hyperref}
\usepackage{colortbl}
\usepackage{xcolor}

\usepackage{pict2e,picture,graphicx}

\makeatletter
\DeclareRobustCommand{\Arrow}[1][]{%
\check@mathfonts
\if\relax\detokenize{#1}\relax
\settowidth{\dimen@}{$\m@th\rightarrow$}%
\else
\setlength{\dimen@}{#1}%
\fi
\sbox\z@{\usefont{U}{lasy}{m}{n}\symbol{41}}%
\begin{picture}(\dimen@,\ht\z@)
\roundcap
\put(\dimexpr\dimen@-.7\wd\z@,0){\usebox\z@}
\put(0,\fontdimen22\textfont2){\line(1,0){\dimen@}}
\end{picture}%
}
\makeatother
\newcommand{\rightarrowshort}{\hspace{.2mm}\scalebox{.8}{\Arrow[.15cm]}\hspace{.2mm}}
\newcommand{\citeas}[1]{\cite{#1}} %%cite author short
\newcommand{\citep}[1]{\cite{#1}} %%cite author short

\newtheorem{theorem}{Theorem}
\newtheorem{observation}{Observation}

\newcounter{exo}
\makeatletter
\newenvironment{example}[1][]
    {\refstepcounter{exo}%
\begin{center}
   \begin{tabular}{|>{\columncolor{gray!10}}p{0.967\textwidth}|}
    \hline \\[-3mm] {\bfseries{Example \theexo}: #1}\\ \hline
    }
    {
    \\[1.5mm] \hline
  \end{tabular}
    \end{center}
    }

\newcommand{\np}{{\sf{NP}}}
\newcommand{\npc}{{\np}-{\sf{complete}}}

\newcommand{\nph}{{\sf{NP-hard}}}
\newcommand{\paranph}{\sf{paraNP}-hard}
\newcommand{\paranphshort}{\sf{paraNP}-hard}

\newcommand{\nphns}{{\np}-{{hardness}}}

\newcommand{\nphshort}{{\sf{NP-hard}}}
\newcommand{\poly}{{\sf{P}}}

\newcommand{\wah}{{\sf{W[1]-hard}}}
\newcommand{\wahshort}{{\sf{W[1]-hard}}}
\newcommand{\wahns}{{\sf{W[1]}}-hardness}

\newcommand{\wac}{{\textsf{W[1]}}-{\sf{complete}}}
\newcommand{\wacns}{{\sf{W[1]}}-{{completeness}}}
\newcommand{\wb}{{\sf{W[2]}}}
\newcommand{\wbh}{\sf{W[2]-hard}}
\newcommand{\wbhshort}{{\sf{{W[2]}-hard}}}
\newcommand{\wbc}{\wb-complete}
\newcommand{\wbhns}{\wb-hardness}

\newcommand{\fpt}{{\sf{FPT}}}

\newcommand{\yesins}{Yes-instance}
\newcommand{\noins}{No-instance}

\newcommand{\bigo}[1]{O(#1)}

\newcommand{\abs}[1]{|#1|}
\newcommand{\edge}[2]{\{#1,#2\}}
\newcommand{\xc}{\mathcal{H}}
\newcommand{\xce}{H}
\newcommand{\xs}{A}
\newcommand{\xse}{a}
\newcommand{\xsize}{\kappa}
\newcommand{\w}{w} % desired committee
\newcommand{\mymid}{:} %%%% seperator in sets
\newcommand{\setmid}{:}
\newcommand{\onlyfull}[1]{}

\newcommand{\prob}[1]{{\sc{#1}}}

\newcommand{\probb}[2]{{\prob{#1}}-{#2}}

\newcommand{\memph}[1]{\emph{#1}}
\newcommand{\score}[3]{\textsf{sc}_{#1}(#2, #3)} %% score of committee #2 in the election #3 wrt the rule #1
\newcommand{\margin}[4]{\textsf{margin}_{#1}(#2\rightarrowshort #3, #4)} %% margin conribution of #2 to #3 in the election #4 wrt the rule #1

\newcommand{\muplus}{\cup}
\newcommand{\vset}{N} %%%vertex set
\newcommand{\eset}{A}  %%% edge set
\newcommand{\vere}{u}
\newcommand{\edee}{e}
\newcommand{\vvset}[1]{N(#1)}

\newcommand{\abmv}{ABMV}
\newcommand{\neighbor}[2]{\Gamma_{#1}(#2)}%{\Gamma_{#1}(#2)} % set of neighbors of vertex #2 in graph #1
\newcommand{\degree}[2]{\textsf{deg}_{#1}(#2)}%{\textsf{deg}_{#1}(#2)} % degree of #2 in graph #1

\newcommand{\yy}{Y} %% a rule satisfies a property
\newcommand{\nn}{N} %% a rule fails a property

\newtheorem{claim}{Claim}
\newtheorem{lemma}{Lemma}
\newtheorem{corollary}{Corollary}

\newcommand{\EPP}[3]
{\begin{center}
{
\renewcommand{\tabcolsep}{0.45mm}
\begin{tabularx}{0.98\columnwidth}{ll}
\toprule
\multicolumn{2}{l}{\small\textsc{#1}} \\ \midrule
{\bf Input:}   & \parbox[t]{0.87\columnwidth}{#2\vspace*{0mm}}  \\
{\bf Question:}    & \parbox[t]{0.87\columnwidth}{#3\vspace*{0mm}} \\ \bottomrule
\end{tabularx}
}
\end{center}}

 \usepackage{geometry}
\begin{document}

\title{%{\bf{Appendix:}}\\
On the Parameterized Complexity of Controlling Approval-Based Multiwinner Voting: Destructive Model \& Sequential Rules%\\ (full version of MFCS 2562)
}

\author{Yongjie Yang}
%\affiliation{
%  \institution{Chair of Economic Theory, Saarland University}
 % \city{Saarb\"{u}cken}
 % \country{Germany}}
%\email{yyongjiecs@gmail.com}

\date{\small{Chair of Economic Theory, Saarland University, Saarb\"{u}cken 66123, Germany \\
yyongjiecs@gmial.com}}

\maketitle

\begin{abstract}
Over the past few years, the (parameterized) complexity landscape of constructive control for many prevalent approval-based multiwinner voting (ABMV) rules has been explored. We expand these results in two directions. First, we study constructive control for sequential Thiele's rules.
Second, we study destructive counterparts of these problems.
Our exploration leads to a comprehensive understanding of the complexity of these problems. Along the way, we also study several interesting axiomatic properties of ABMV rules, and obtain generic results for rules fulfilling these properties. In particular, we show that for many rules satisfying these properties, election control problems are generally hard to solve from a parameterized complexity point of view, even when restricted to certain special cases.
\medskip

\end{abstract}
%%%%%%%%%%%%%%%%%%%%%%%%%%%%%%%%%%%%%%%%%%%%%%%%%%%%%%%%%%%%%%%%%%%%%%%%%%%%%%%%%%%%%%%%%%%%%%%%%%%%%%%%%

\section{Introduction}
Voting has been commonly recognized as a prominent mechanism for collective decision-making. Since its application has been expanded from traditional scenarios (e.g., political elections) to a much broader range of areas where voting with a huge number of voters or candidates may be involved, analyzing voting rules from a computer science perspective has become vital. In particular,  to what extent a voting rule  resists strategic voting has been regarded as an important factor to judge the suitability and quality of voting rules in practice, and computer science promises to offer us a fine-grained answer to this question~\citep{Bartholdi92howhard}. A voting rule can be either immune or susceptible to a particular type of strategic voting. Immunity is the highest level of resistance of voting rules to strategic voting.
If a rule is susceptible to a strategic type, investigating whether it is easy (polynomial-time solvable) or hard (\nph) to successfully perform the corresponding strategic action is the next step to measure its resistance degree. On top of that, parameterized complexity theory offers us a more subtle classification by identifying fixed-parameter tractable (\fpt) problems and fixed-parameter intractable ({\wah}, {\wbh},  etc.) problems from those proven to be {\nph}. From a purely theoretical point of view, higher complexity implies higher degree of resistance.\footnote{There are also experimental studies of strategic voting problems (see, e.g.,~\cite{DBLP:journals/jcss/ErdelyiFRS15a,DBLP:conf/aaai/LoreggiaNRVW14}),
which provide complementary insight towards understanding whether hard strategic problems can be solved efficiently in practice.}

Arguably, investigating the complexity of many strategic voting problems for single-winner voting rules had largely dominated the early development (during 1990--2010) of Computational Social Choice. Thanks to the tremendous effort of many researchers, an almost complete landscape of the complexity of related problems has emerged~\cite{BaumeisterR2016,DBLP:journals/corr/abs-2207-00710,handbookofcomsoc2016Cha7FR}. However, the past few years have witnessed a shift of research focus from single-winner voting to multiwinner voting (see, e.g.,~\cite{DBLP:conf/aaai/BrillIM022,DBLP:conf/aaai/CaragiannisF22,DBLP:conf/ijcai/DoHL022,DBLP:conf/ijcai/SornatWX22}).
Towards achieving a comprehensive understanding of the complexity of related problems for multiwinner voting rules, we study eight standard election control problems for important approval-based multiwinner voting (\abmv) rules. These problems model the scenario where a powerful external agent (i.e., the chair of an election) aims to either make a set of distinguished candidates be included in all winning $k$-committees (constructive) or make none of them be contained in any winning $k$-committees (destructive) by performing a certain type of strategic action (e.g., adding or deleting a limited number of voters or candidates).
To make our results as general as possible, we also study several axiomatic properties of {\abmv} rules, which might be of independent interest, and derive results applying to rules satisfying these properties.

\begin{table}
\caption{
A summary of the complexity of election control problems for ABMV rules.
Our results for sequential rules may apply to many other sequential Thiele's rules satisfying certain properties not shown in the table (see corresponding theorems for the details).
Here,~$J$ is the set of distinguished candidates,~$m$ is the number of candidates,~$n$ is the number of votes,~$k$ is the size of the winning committee, and~$\ell$ is the solution size.
Parameters with respect to which parameterized complexity results hold are shown as superscripts.
Note that {\wahns}/{\wbhns} of a problem with respect to a combined parameter~$\kappa+\kappa'$ implies that the problem is {\wah}/{\wbh} when parameterized by~$\kappa$ and~$\kappa'$ separately. For other results concerning constructive control for ABMV rules not shown here we refer to~\protect\cite{DBLP:conf/ijcai/Yang19}.}
\label{tab-results-summary}
\renewcommand{\tabcolsep}{2.34mm}
\begin{center}
\footnotesize{
\begin{tabular}{|l|l|l|l|l|}\toprule
       &    {\prob{CCAV}} & {\prob{CCDV}}  & {\prob{CCAC}} & {\prob{CCDC}} \\ \toprule

seqABCCV
&{\wahshort}$^{\ell}$, {\paranphshort}$^k$
& {\wbhshort}$^{\ell}$ ($k=1$, \protect\cite{DBLP:journals/tcs/LiuFZL09})
& {\wbhshort}$^{k+\ell}$
& {\nphshort} \\

seqPAV
%$\omega$-Thiele
&  ($\abs{J}=1$, Thm.~\ref{ccav-seqpav-nph}) % {\wahshort}$^{\ell}$ \protect\cite{DBLP:journals/tcs/LiuFZL09}
&  {\paranphshort}$^k$  ($\abs{J}=1$, Thm.~\ref{ccdv-seqpav-nph})
& ($\abs{J}=1$, Thm.~\ref{thm-ccac-seqThiele-wbh})
&($\abs{J}=1$, Thm.~\ref{thm-ccdc-seqpav-seqabccv-np-hard})  \\  \bottomrule%\midrule
%&
%$\abs{J}=1$, Thm.~\ref{ccav-seqpav-nph})
%&% $\triangle_{\text{V}}\leq 3, \triangle_{\text{C}}\leq 2$, Thm.~\ref{ccdv-seqpav-nph})
%& $\triangle_{\text{C}}\leq 3$, Cor.~\ref{cor-ccac-thiele-np-hard})
%& %$\triangle_{\text{C}}\leq 3$, Thm.~\ref{thm-ccdc-seqpav-seqabccv-np-hard})
%\\ \bottomrule%\midrule

%seqABCCV
%&  {\wahshort} ($\abs{J}=1$,
%&{\nphshort} ($k=\bigo{1}, \abs{J}=1,$
%&{\nphshort} ($\abs{J}=1, \triangle_{\text{V}}\leq 4,$
%& {\nphshort} ($\abs{J}=1, \triangle_{\text{V}}\leq 4,$  \\

%& $k=\bigo{1}$, Thm.~\ref{ccav-seqpav-nph})
%&$\triangle_{\text{V}}\leq 3, \triangle_{\text{C}}\leq 2$, Thm.~\ref{ccdv-seqpav-nph})
%&$\triangle_{\text{C}}\leq 3$, Thm.~\ref{thm-ccac-seqThiele-wbh})
%& $\triangle_{\text{C}}\leq 3$, Thm.~\ref{thm-ccdc-seqpav-seqabccv-np-hard})  \\ \bottomrule
\end{tabular}

%\end{center}

\renewcommand{\tabcolsep}{0.81mm}
%\renewcommand\arraystretch{0.95}
%\begin{center}
\begin{tabular}{|l|l|l|l|l|}\toprule
          &   {\prob{DCAV}} & {\prob{DCDV}}  & {\prob{DCAC}} & {\prob{DCDC}} \\ \toprule

AV
&  {\wbhshort}$^{\ell}$ ($\abs{J}=1$, Thm.~\ref{thm-dcav-many-rules-wbh})
& {{\wbhshort}$^{n-\ell}$} ($\abs{J}=1$, Thm.~\ref{thm-dcdv-save-wbh-remaining-votes})
& {\poly} \cite{DBLP:journals/jair/MeirPRZ08}
& immune (Thm.~\ref{thm-chernoff-rule-immue-dcdc}) \\

& $\alpha$-efficient \&  neutral
&  {{\wahshort}$^{k+\ell}$} ($\abs{J}=1$, Thm.~\ref{thm-dcdv-av-wah-k-ell})
&
& Chernoff\\ \cline{1-1} \cline{4-5}

SAV
& {\wahshort}$^{k+\ell}$ ($\abs{J}=1$, Thm.~\ref{thm-dcav-many-rules-wbh-k-ell})
& AV-uniform
& {\nphshort} ($\abs{J}=1$,
& {\wbhshort}$^{\ell}$  \\%, \triangle_{\text{V}}\leq 4$  \\

NSAV
&  AV-uniform
&
& Cor.~\ref{cor-dcac-nsav-np-hard}, Thm.~\ref{thm-dcac-sav-np-hard})
& ($\abs{J}=1$, Thm.~\ref{thm-dcdc-sav-np-hard}\onlyfull{, Cor.~\ref{cor-dcdc-nsav-np-hard}}) %  $\triangle_{\text{C}}\leq 3$ Thm.~\ref{thm-dcdc-sav-np-hard})
\\ \midrule

seqABCCV
& {\wahshort}$^{k}$\ \ \ ($\abs{J}=1$, Thm.~\ref{thm-dcav-seqabccv-wa-hard})
& {\wbhshort}$^{\ell}$ (Thm.~\ref{thm-dcdv-seq-wbh})
& {\wahshort}$^{k+\ell}$
&{\wahshort}$^{m-\ell, k}$  \\ \cline{1-2}

seqPAV
&  {\wahshort}$^{k, \ell}$ ($\abs{J}=1$, Thms.~\ref{thm-dcav-seqpav-nph}, \ref{thm-dcav-seqabccv-wa-hard})
&
& ($\abs{J}=1$, Thm.~\ref{thm-dcac-thiele-rule})
& ($\abs{J}=1$,  Thm.~\ref{thm-dcdc-seqThiele-wah}) \\  \midrule
%seqPAV
%& {\wahshort} ($\abs{J}=1, V=\emptyset,$
%&
%&{\wahshort} ($\abs{J}=1$, Thm.~\ref{thm-dcac-thiele-rule})
%& {\wahshort}$^{k+m-\ell}$ (Thm.~\ref{thm-dcdc-seqThiele-wah}) \\
%
%& $\triangle_{\text{C}}\leq 2$, Thm.~\ref{thm-dcav-seqpav-nph})
%&
%&
%&{\nphshort} ($\triangle_{\text{V}}, \triangle_{\text{C}}\leq 3$, Cor.~\ref{cor-dcdc-seqThiele-np-hard-voters-3-canddiates-3})  \\ \midrule

%&  {\wbhshort} ($k=m-1, V=\emptyset,$
%&
%&
%&\\
%
%&$\abs{J}=1$, Thm.~\ref{thm-dcav-many-rules-wbh})
%&
%&
%& \\ \midrule

ABCCV/PAV & \multicolumn{4}{l|}{{\wahshort}$^{k}$, {\paranphshort}$^{\ell}$ ($\abs{J}=1$, Thm.~\ref{thm-abccv-very-speical-case-np-hard})}
%, {\wbhshort}$^{k}$ (\protect\cite{DBLP:journals/jair/BetzlerSU13,DBLP:conf/atal/MisraNS15})}
\\  \hline

MAV & \multicolumn{4}{l|}{{\wbhshort}$^{k}$, {\paranphshort}$^{\ell}$ ($\abs{J}=1$, Thm.~\ref{thm-mav-very-speical-case-np-hard})}
%, {\wahshort}$^{k}$ (\protect\cite{DBLP:conf/atal/AzizGGMMW15})}
\\ \bottomrule%\midrule
\end{tabular}
}
\end{center}
\end{table}

\subsection{Related Works and Our Main Contributions}

The complexity of strategic voting problems was first studied by Bartholdi, Tovey, and Trick~\citeas{BARTHOLDI89,Bartholdi92howhard}, who introduced several constructive control problems. The destructive counterparts of these problems were first studied by Hemaspaandra, Hemaspaandra, and Rothe~\citeas{DBLP:journals/ai/HemaspaandraHR07}. The number of papers investigating the (parameterized) complexity of these problems for single-winner voting rules is enormous (see, e.g.,~\cite{DBLP:journals/jair/FaliszewskiHH11,DBLP:journals/mlq/HemaspaandraHR09,DBLP:journals/ai/NevelingR21}). We refer to the book chapters~\cite{BaumeisterR2016,handbookofcomsoc2016Cha7FR} for a consultation for results by 2016, and refer to~\cite{DBLP:journals/aamas/ErdelyiNRRYZ21,DBLP:journals/ai/NevelingR21,DBLP:conf/atal/Yang17,DBLP:conf/ecai/000120} for some recent results.

Meir~\citeas{DBLP:journals/jair/MeirPRZ08} initiated the study of the complexity of control problems for multiwinner voting rules where the main focus are ranking-based rules. However, we focus on {\abmv} rules. Besides, in their model, the external agent derives utilities from candidates and attempts to achieve a winning~$k$-committee yielding a total utility (sum of utilities of committee members) exceeding a given threshold~$s$. For other strategic problems for ranking-based multiwinner voting rules, we refer to~\citep{DBLP:journals/ai/BredereckFKNST21,DBLP:journals/toct/BredereckFNT21,DBLP:journals/aamas/BredereckKN21} and references therein.

Yang~\citeas{DBLP:conf/ijcai/Yang19} studied the complexity of constructive control by adding/deleting votes/candidates for many {\abmv} rules including approval voting (AV), satisfaction approval voting (SAV), net-SAV (NSAV), approval-based Chamberlin-Courant voting (ABCCV), proportional approval voting (PAV), and minimax approval voting (MAV). Magiera~\citeas{Magierphd2020} complemented this work by investigating destructive control by deleting votes/candidates for AV and SAV, and showed {\nphns} for these problems in the special case where there is only one distinguished candidate.

Other strategic problems for {\abmv} rules have also  been studied recently. Faliszewski, Skowron, and Talmon~\citeas{DBLP:conf/atal/FaliszewskiST17} studied bribery problems for {\abmv} rules, where the goal is to make a distinguished candidate be contained in a winning $k$-committee by performing a limited number of modification operations. Later, Yang~\citeas{DBLP:conf/atal/000120} studied the destructive counterparts of these problems. Faliszewski, Gawron, and Kusek~\citeas{DBLP:conf/eumas/FaliszewskiGK22}, and Gawron and Faliszewski~\citeas{DBLP:conf/aldt/GawronF19} studied the robustness of {\abmv} rules, which can be considered as variants of bribery problems. Markakis and Papasotiropoulos~\citeas{DBLP:conf/ijcai/MarkakisP21} studied control under conditional approval voting.

Our main contributions are summarized as follows.
\begin{enumerate}
\item[(1)] We first study constructive control for numerous sequential $\omega$-Thiele's rules containing the two prevalent {\abmv} rules \mbox{seqABCCV} and seqPAV, and then we study destructive control for all the {\abmv} rules mentioned above.
\item[(2)] We examine several axiomatic properties of {\abmv} rules and derive several general complexity results for {\abmv} rules satisfying these properties.
\item[(3)] We offer a complete landscape of the complexity of the eight standard control problems for the aforementioned {\abmv} rules and provide a comprehensive understanding of their parameterized complexity with respect to mainly the natural parameters such as the solution size and the size of winning committees.
Table~\ref{tab-results-summary} summarizes our main results.
\end{enumerate}

We remark that all our reductions established in the paper run in polynomial time, and all the eight standard control problems for the aforementioned concrete rules are obviously in~{\np}. As a consequence, if a problem is shown to be {\wah} or {\wbh}, it implies that the problem is also {\npc}. It should be remarked that all our hardness results for constructive control can be trivially adapted for establishing he same hardness for the  corresponding utility-based control problems formulated in~\cite{DBLP:journals/jair/MeirPRZ08}: just set the utility of each distinguished candidate to be~$1$, set that of each other candidate to be~$0$, and set the threshold~$s$ to be the number of distinguished candidates.

\section{Preliminaries}
\label{sec-preliminaries}
We assume the reader is familiar with the basics in graph theory and (parameterized) complexity~\cite{DBLP:books/sp/CyganFKLMPPS15,DBLP:conf/lata/Downey12,DBLP:journals/interfaces/Tovey02,Douglas2000}.

For an integer~$i$,~$[i]$ denotes the set of all positive integers at most~$i$. For a set~$S$, we use~$\overrightarrow{S}$ to denote an arbitrary linear order over~$S$.

\subsection{Approval-Based Multiwinner Voting}
An election is a tuple~$(C, V)$ consisting of a set~$C$ of candidates and a multiset~$V$ of votes cast by a set of voters. Each vote~$v\in V$ is defined as a subset of~$C$.
Throughout the paper, the terms vote and voter are used interchangeably.
We say that a vote~$v$ approves a candidate~$c$ if~$c\in v$. For $C'\subseteq C$, we use~$V_{C'}$ to denote the multiset of votes obtained from~$V$ by removing $C\setminus C'$ from every vote in~$V$. Therefore, $(C', V_{C'})$ is the election $(C, V)$ restricted to~$C'$.
For notational brevity, we usually write $(C', V)$ for $(C', V_{C'})$.
For a candidate~$c\in C$,~$V(c)$ is the multiset of votes in~$V$ approving~$c$.
A~$k$-set is a set  of cardinality~$k$. A committee (resp.\ $k$-committee) refers to a subset (resp.\ $k$-subset) of candidates.

An ABMV rule~$\varphi$ maps each election $(C, V)$ and an integer~$k\leq \abs{C}$ to a collection $\varphi(C, V, k)$ of~$k$-committees of~$C$, which are called {{winning~$k$-committees}} of~$\varphi$ at~$(C, V)$.\footnote{In practice, when there are multiple winning $k$-committees, a tie-breaking rule is used to select one from them.}

Now we elaborate on ABMV rules studied in the paper. Under these rules, each~$k$-committee receives a score based on the votes, and the rules select winning $k$-committees maximizing or minimizing this score. We first consider additive rules~\cite{Kilgour2010,DBLP:conf/ijcai/YangW18}, where the score of a committee is the sum of the scores of its members, and these rules select $k$-committees with the maximum score. Three prevalent {\abmv} rules are formally defined  as follows.

\begin{description}
\item[AV] The AV score of a candidate~$c\in C$ is the number of votes approving~$c$, and the AV score of a committee~$w\subseteq C$ is $\sum_{v\in V}\abs{v\cap \w}$. The rule selects $k$-committees with the maximum AV score among all $k$-committees.

\item[SAV] Each candidate~$c\in C$ receives~$\frac{1}{\abs{v}}$ points from each vote~$v\in V$ approving~$c$. The SAV score of a committee~$w\subseteq C$ is then $\sum_{v\in V, v\neq\emptyset} \frac{\abs{w\cap v}}{\abs{v}}$. The rule selects $k$-committees with the maximum SAV score among all $k$-committees.

\item[NSAV] Each candidate~$c\in C$ receives $\frac{1}{\abs{v}}$ points from each vote~$v\in V$ approving~$c$, and receives  $-\frac{1}{\abs{C\setminus v}}$ points from each vote~$v\in V$ disapproving~$c$. The NSAV score of a committee~$w\subseteq C$ is then $\sum_{v\in V, v\neq\emptyset} \frac{\abs{w\cap v}}{\abs{v}}-\sum_{v\in V, v\neq C} \frac{\abs{w\setminus v}}{\abs{C\setminus v}}$. The rule selects $k$-committees with the maximum NSAV score among all $k$-committees.
\end{description}

Now we give definitions of another group of rules where the score of each committee is nonlinearly determined by the scores of its members.

\begin{description}
\item[ABCCV] A voter is satisfied with a committee if this committee includes at least one of her approved candidates. The ABCCV score of a committee is the number of voters who are satisfied with the committee, and winning $k$-committees are those with the maximum ABCCV score.
\end{description}

ABCCV is a variant of the original rules proposed by Chamberlin and Courant~\cite{ChamberlinC1983APSR10.2307/1957270}.
%, and was suggested by Thiele~\cite{Thiele1895}.

\begin{description}
\item [PAV] The PAV score of a committee~$\w\subseteq C$ is defined as \[\sum_{v\in V, v\cap \w\neq \emptyset}\left(\sum_{i=1}^{\abs{v\cap \w}}\frac{1}{i}\right).\] Winning $k$-committees are those with the maximum PAV score.
%\onlyfull{PAV fulfills several proportional properties~\cite{DBLP:journals/scw/AzizBCEFW17}.}

\item[MAV] The Hamming distance between two subsets~$\w\subseteq C$ and~$v\subseteq C$ is $\abs{\w\setminus v}+\abs{v\setminus \w}$. The score of a committee~$\w$ is the maximum Hamming distance between~$\w$ and the votes, i.e., $\max_{v\in V} (\abs{\w\setminus v}+\abs{v\setminus \w})$. Winning $k$-committees are those having the minimum MAV score.
\end{description}

PAV was first studied by Thiele~\citeas{Thiele1895}, and MAV was proposed by Brams~\citeas{Bramsminimaxapproval2007}.

It should be noted that calculating a winning $k$-committee for ABCCV, PAV, and MAV is {\nph}~\cite{DBLP:conf/atal/AzizGGMMW15,LeGrand2004Tereport,DBLP:journals/scw/ProcacciaRZ08}.
%, standing in contrast to the polynomial-time solvability for additive rules.%
%We also remark that every rule has its merits and its drawbacks
%%\medskip
%%

We also derive results for a class of rules named $\omega$-Thiele's rules which contain some of the above defined rules.

\begin{description}
\item[$\omega$-Thiele's rules] Each Thiele's rule is characterized by a function $\omega: \mathbb{N}\rightarrow \mathbb{R}$ so that $\omega(0)=0$ and $\omega(i+1)\geq \omega(i)$ for all nonnegative integers~$i$. The score of a committee $\w\subseteq C$ is defined as $\sum_{v\in V}\omega(\abs{v\cap w})$. The rule selects $k$-committees with the maximum score.\footnote{Thiele's rules are also studied under other names such as weighted PAV, generalized approval procedures, etc.\ (see, e.g.,~\citep{DBLP:journals/scw/AzizBCEFW17,Kilgour2012}).}
\end{description}

Obviously, AV is the $\omega$-Thiele rule where $\omega(i)=i$ for all $i>0$, ABCCV is the $\omega$-Thiele's rule where $\omega(i)=1$ for all $i>0$, and PAV is the $\omega$-Thiele's rule where $\omega(i)=\sum_{j=1}^i 1/j$ for all $i>0$. Many of our results apply to a subclass of Thiele's rules satisfying some properties.

For each integer~$i$, let $\Delta_{\omega}(i)=\omega(i+1)-\omega(i)$. .
For an ABMV rule~$\varphi$ defined above, an election $E=(C, V)$, and a committee $\w\subseteq C$, we use $\score{\varphi}{\w}{E}$ to denote the~$\varphi$ score of~$\w$ in~$E$. The~$\varphi$ margin contribution of a candidate $c\in C$ to~$\w$ with respect to the election~$E$ is defined as
\[\margin{\varphi}{c}{\w}{E}=\score{\varphi}{\w\cup \{c\}}{E}-\score{\varphi}{\w}{E}.\]

In the paper, we also study sequential versions of ABCCV  and PAV which fall into the category of sequential Thiele's rules defined as follows.

\begin{description}
\item[Sequential $\omega$-Thiele's rules] Let~$\varphi$ be an $\omega$-Thiele's rule. Sequential~$\varphi$ (seq$\varphi$) determines the winning $k$-committee by including its members iteratively. Initially, we have  $\w=\emptyset$. In the next iteration, we add a candidate from $C\setminus \w$ into~$\w$ who has the maximum~$\varphi$ margin contribution to $\w$. If multiple candidates have the same margin contribution, a tie-breaking scheme is used to select the one added into~$\w$. In this paper, we assume that ties are broken by a linear order over~$C$. The procedure terminates when~$\w$ contains exactly~$k$ candidates.
\end{description}

It should be noted that seqAV is exactly AV equipped with a linear tie-breaking order.
For further detailed discussions on the above defined rules,  we refer to the recent comprehensive survey by Lackner and Skowron~\citeas{DBLP:series/sbis/LacknerS23}.

The following notions are used in several of our algorithms concerning sequential Thiele's rules.
For an $\omega$-Thiele's rule~$\varphi$, an election~$E=(C, V)$, a linear order~$\rhd$ over~$C$, and a positive integer $i\leq k\leq \abs{C}$, we use~$\w^i_{\omega, \rhd, E}$ to denote the $i$-th candidate included in the winning $k$-committee of seq$\varphi$ at~$E$ when ties are broken by~$\rhd$. Moreover, let
\[\w^{\leq i}_{\omega, \rhd, E}=\{\w^j_{\omega, \rhd, E} \setmid j\leq i\}.\] When it is clear from the context which~$\omega$ and~$\rhd$ are considered, we simply write~$\w^i_E$ for~$\w^i_{\omega, \rhd, E}$, and write~$\w^{\leq i}_E$ for~$\w^{\leq i}_{\omega, \rhd, E}$.

\subsection{Some Axiomatic Properties of {\abmv} Rules}
In this section, we provide definitions of some axiomatic properties of {\abmv} rules used to establish our complexity results.

\begin{description}
    \item[$\alpha$-efficiency]
    An {\abmv} rule~$\varphi$ is $\alpha$-efficient if for every election $(C, V)$ and every integer $k$, if $C_0=\{c\in C \setmid V(c)=\emptyset\}\neq \emptyset$ and $\abs{C\setminus C_0}\geq k$, it holds that none of~$C_0$ is contained in any winning $k$-committee of~$\varphi$ at $(C, V)$.\footnote{Lackner and Skowron~\protect\citeas{DBLP:journals/jet/LacknerS21} studied the notion of weak efficiency  which requests that if a winning committee~$\w$ contains some candidate~$c$ not approved by any votes, then replacing~$c$ with any candidate not in~$\w$ yields another winning committee.
Peters~\citeas{petersphdthesis} defined a similar notion with the same name for resolute rules, i.e., rules that always return exactly one winning committee. Our notion equals that of Peters when restricted to resolute rules. To avoid confusion, we use the name $\alpha$-efficiency. Indeed, for irresolute rules, $\alpha$-efficiency and weak efficiency (defined in \cite{DBLP:journals/jet/LacknerS21}) are independent.}

\item[Neutrality] (\protect\cite{DBLP:series/sbis/LacknerS23}) An {\abmv} rule is neutral if candidates are treated equally. More precisely, let $E=(C, V)$ be an election and let $k\leq \abs{C}$ be a positive integer. For a bijection $f: C\rightarrow C$, let~$E_f$ be the election obtained from~$E$ by changing each $v\in V$ into $\{f(c) \setmid c\in v\}$. A rule~$\varphi$ is neutral if for every election $E=(C, V)$, every positive integer $k\leq \abs{C}$, and every bijection $f: C\rightarrow C$, it holds that a $k$-committee~$w$ of~$C$ is a winning $k$-committee of~$\varphi$ at~$E$ if and only if  $\{f(c)\setmid c\in w\}$ is a winning $k$-committee of~$\varphi$ at~$E_f$.

\item[AV-Uniform]  An {\abmv} rule~$\varphi$ is AV-uniform if for every election $(C, V)$ and every integer $k\leq \abs{C}$ such that every vote in~$V$ approves the same number of candidates, it holds that the winning $k$-committee of~$\varphi$ at $(C, V)$ coincide with those of AV at $(C, V)$, i.e., $\varphi(C,V,k)=\text{AV}(C, V, k)$.
\end{description}

Clearly, AV, SAV, and NSAV are AV-uniform, and none of other concrete rules studied in the paper are AV-uniform.
Additionally, it is easy to see that AV, SAV, NSAV, and PAV are $\alpha$-efficient and neutral.
However, ABCCV, seqABCCV, and MAV are not~$\alpha$-efficient\footnote{However, ABCCV and MAV satisfy weak efficiency defined by Lackner and Skowron~\citeas{DBLP:journals/jet/LacknerS21}.}, as illustrated in Example~\ref{ex-b}.

\begin{example}[ABCCV, seqABCCV, and MAV are not~$\alpha$-efficient]
\label{ex-b}
Consider an election with three candidates~$a$,~$b$, and~$c$, and with one vote which approves exactly~$a$ and~$b$. All $2$-committees are ABCCV winning $2$-committees. In addition, with the tie-breaking order $c\rhd a\rhd b$, $\{a, c\}$ is a seqABCCV winning $2$-committee. Therefore, neither ABCCV nor its sequential version is $\alpha$-efficient. To check that MAV is not $\alpha$-efficient, observe that any $1$-committee is an winning $1$-committee.
\end{example}

Observe that seqAV and seqPAV are $\alpha$-efficient, but they are neither neutral nor AV-uniform. As a matter of fact, it is easy to see that an $\omega$-Thiele's rule satisfies $\alpha$-efficiency if and only if $\triangle_{\omega}(i)>0$ for all $i\geq 0$, and none of the sequential $\omega$-Thiele's rules is neutral when ties are broken by a linear order.

Now we elaborate on the definition of the Chernoff property.
Let~$X$ be a set. A choice function~$f: 2^X\rightarrow 2^X$ satisfies the Chernoff property if for every~$Y\subseteq X$, it holds that $(f(X)\cap Y)\subseteq f(Y)\subseteq Y$~\cite{chernoffeconometrica1954}\footnote{We note that in some literature, a choice function refers to a function which returns exactly one single element from~$X$, and the function defined here is called a choice correspondence.}. This property has been studied under several other names (see, e.g.,~\cite{sen1971TRES,PetesPtheorydecision2021}). We extend this properly to {\abmv} rules in a natural way.

\begin{description}
\item[Chernoff] An {\abmv} rule~$\varphi$ satisfies the Chernoff property if for every election $(C, V)$, every nonnegative integer $k\leq \abs{C}$, and every $C'\subseteq C$ such that $\abs{C'}\geq k$ the following condition holds: for every $\w\in \varphi(C,V, k)$ there exists $\w'\in \varphi(C',V_{C'}, k)$ such that $(\w\cap C')\subseteq \w'$. Generally speaking, this property says that for every winning $k$-committee~$\w$ of the election $(C, V)$, if we consider the election restricted to some $C'\subseteq C$ then there exists at least one winning $k$-committee of the subelection which contains all candidates in~$\w$ that are also contained in~$C'$.
\end{description}

AV and seqAV clearly satisfy the Chernoff property.
However, none of the other concrete rules studied in the paper satisfies the Chernoff property.
It is interesting to see if there are other natural {\abmv} rules satisfying the Chernoff property.

A summary of whether a particular rule satisfies a particular property defined above is given in Table~\ref{tab-properties}.

\begin{table*}
\caption{A summary of axiomatic properties of {\abmv} rules. Here, ``\yy''/``\nn'' in an entry means that the corresponding rule satisfies/dissatisfies the property.}
\label{tab-properties}
    \centering{
    \begin{tabular}{|l|c|c|c|c|}\hline
         & $\alpha$-efficient & neutral & AV-uniform & Chernoff  \\ \hline

     AV  & \yy & \yy & \yy &\yy \\ \hline

     SAV  & \yy & \yy & \yy &\nn \\ \hline

    NSAV  & \yy & \yy & \yy &\nn \\ \hline

    PAV  & \yy & \yy & \nn &\nn \\ \hline

ABCCV  & \nn & \yy & \nn &\nn \\ \hline

 MAV  & \nn & \yy & \nn &\nn \\ \hline

 seqPAV  & \yy & \nn & \nn &\nn \\  \hline

seqABCCV  & \nn & \nn & \nn &\nn \\ \hline
    \end{tabular}
    }
\end{table*}

\subsection{Problem Formulations}

Now we extend the definitions of several standard single-winner election control problems to multiwinner voting. We first give the definitions of destructive control problems. These problems model the scenario where some external agent (e.g., the election chair) aims to make some distinguished candidates the winners by reconstructing the election.

\EPP
{Destructive Control by Adding Voters for~$\varphi$ (\probb{DCAV}{$\varphi$})}
{An election~$(C, V\muplus U)$, an integer~$k\leq \abs{C}$, a nonempty subset~$J\subseteq C$, and an integer~$\ell\leq \abs{U}$.}
{Is there $U'\subseteq U$ such that~$|U'|\leq \ell$ and none of~$J$ belongs to any winning $k$-committee of~$\varphi$ at $(C, V\muplus U')$?}

In the above definition, votes in~$V$ and votes in~$U$ are respectively called registered votes and unregistered votes.

\EPP
{Destructive Control by Deleting Voters for~$\varphi$ (\probb{DCDV}{$\varphi$})}
{An election~$(C, V)$, an integer~$k\leq \abs{C}$, a nonempty subset~$J\subseteq C$, and an integer~$\ell\leq \abs{V}$.}
{Is there $V'\subseteq V$ such that~$\abs{V'}\leq \ell$ and none of~$J$ belongs to any winning $k$-committee of~$\varphi$ at $(C, V\setminus V')$?}

\EPP
{Destructive Control by Adding Candidates for~$\varphi$ (\probb{DCAC}{$\varphi$})}
{An election~$(C\muplus D, V)$, an integer~$k\leq \abs{C}$, a nonempty subset $J\subseteq C$, and an integer~$\ell\leq \abs{D}$.}
{Is there $D'\subseteq D$ such that $\abs{D'}\leq \ell$ and none of~$J$ belongs to any winning $k$-committee of~$\varphi$ at $(C\muplus D', V)$?}

In the above definitions, we call candidates in~$C$ and~$D$ respectively registered candidates and unregistered candidates.

\EPP
{Destructive Control by Deleting Candidates for~$\varphi$ (\probb{DCDC}{$\varphi$})}
{An election~$(C, V)$, an integer~$k\leq \abs{C}$, a nonempty subset~$J\subseteq C$, and an integer~$\ell\leq \abs{C}-k$.}
{Is there $C'\subseteq C\setminus J$ such that~$\abs{C\setminus C'}\geq k$ and none of~$J$ belongs to any winning $k$-committee of~$\varphi$ at $(C\setminus C', V)$?}

In the above definitions, candidates in~$J$ are called distinguished candidates.
The above problems model the scenario where an external agent aims to prevent any of the distinguished candidates from being contained in {any winning $k$-committee (other than one particular winning $k$-committee).
This is realistic when the external agent is unaware of the tie-breaking rule, or when a randomized tie-breaking rule is used. In this case, the external agent may want to make sure that her goal is reached without taking any risk.

In {\probb{DCDC}{$\varphi$}}, we do not allow to delete the distinguished candidates. This captures the setting where the external agent (chair) does not want the public to easily identify her distinguished (disliked) candidates, thereby be criticized for abusing power. Besides, this is consistent with the original definition of {\prob{DCDC}} for single-winner voting rules.

We also study constructive counterparts of the above problems. More concretely, for each $X\in \{\text{AV}, \text{DV}, \text{AC}, \text{DC}\}$, the {\probb{CCX}{$\varphi$}} problem takes the same input as {\probb{DCX}{$\varphi$}}, and asks if it is possible to make all distinguished candidates contained in all winning $k$-committees by performing the corresponding control operation. We refer to~\cite{DBLP:conf/ijcai/Yang19} for formal definitions of these problems.

Note that when $\abs{J}=k=1$, the above destructive control problems are exactly the nonunique-winner models of the same-named problems for single-winner voting rules, and the constructive control problems are the unique-winner models of the same-named problems for single-winner voting rules that have been extensively and intensively studied in the literature\onlyfull{~\cite{Bartholdi92howhard,DBLP:journals/jair/FaliszewskiHHR09,handbookofcomsoc2016Cha7FR,DBLP:conf/atal/Yang17}}.

A rule~$\varphi$ is immune to a destructive \onlyfull{(resp.\ constructive) }control type (DCAV, DCDV, DCAC, DCDC), if the external agent cannot exert the corresponding control operation successfully, in the sense that~$J$ cannot be turned from intersecting\onlyfull{ (resp.\ not being contained in all winning $k$-committees)} at least one winning $k$-committee into being entirely disjoint from any winning $k$-committees\onlyfull{ (resp.\ contained in all winning $k$-committees)} by performing the control operation. The immunity of~$\varphi$ to a constructive type is defined in a similar spirit. Precisely, we say that a rule~$\varphi$ is immune to a  constructive control type (CCAV, CCDV, CCAC, CCDC), if the external agent cannot turn~$J$ from not being contained in all winning $k$-committees into being  contained in all winning $k$-committees by performing the corresponding control operation.

We also study a problem which is a special case of all the above destructive control problems. An analogous special case of constructive control has been considered in~\cite{DBLP:conf/ijcai/Yang19}.

\EPP{Nonwinning Testing for~$\varphi$}
{An election $(C, V)$, a nonnegative integer~$k\leq \abs{C}$, and a candidate $p\in C$.}
{Is~$p$ not contained in any winning $k$-committee of $\varphi$ at~$(C, V)$?}

\subsection{Useful Hardness Problems}
Our hardness results are based on reductions from the following problems.

Let $G=(\vset, \eset)$ be a graph where~$\vset$ is the set of vertices and~$\eset$ is the set of edges of~$G$.
We say that a vertex~$\vere\in \vset$ covers an edge, or say that~$\vere$ is incident to this edge, if~$\vere$ is one of the two endpoints of the edge. The edges that are covered by a subset $S\subseteq N$ of vertices refer to the edges covered by at least one vertex in~$S$.
A {\it{vertex cover}} of~$G$ is a subset of vertices covering all edges of~$G$. In other words,~$S\subseteq \vset$ is a vertex cover of~$G$ if after removing all vertices in~$S$ there are no edges. A subset~$S\subseteq N$ is an independent set of~$G$ if $\vset\setminus S$ is a vertex cover of~$G$.
A clique of~$G$ is a subset of vertices so that there is an edge between every two vertices in the subset.
For a vertex $\vere\in \vset$, let $\neighbor{G}{\vere}=\{\vere'\in \vset \setmid \edge{\vere}{\vere'}\in \eset\}$ be the set of neighbors of~$\vere$ in~$G$. For $\vset'\subseteq \vset$, let $\neighbor{G}{\vset'}=(\bigcup_{\vere\in \vset'}\neighbor{G}{\vere})\setminus \vset'$. The degree of a vertex~$\vere\in \vset$ is defined as $\degree{G}{\vere}=\abs{\neighbor{G}{\vere}}$.
A regular graph is a graph where all vertices have the same degree. More specifically, a $d$-regular graph is a graph where every vertex has degree~$d$.

\EPP
{Vertex Cover}
{A graph~$G=(N, A)$, and an integer~$\kappa$.}
{Does~$G$ have a vertex cover of size~$\kappa$?}

It is well-known that {\prob{Vertex Cover}} is {\npc}, and it remains so even when restricted to $3$-regular graphs~\cite{DBLP:journals/tcs/GareyJS76,DBLP:journals/jct/Mohar01}.

\EPP
{Independent Set}
{A graph~$G=(N, A)$, and an integer~$\kappa$.}
{Does~$G$ have an independent set of size~$\kappa$?}

\EPP
{Clique}
{A graph $G=(\vset, \eset)$ and an integer~$\kappa$.}
{Does~$G$ have a clique of size~$\kappa$?}

It is known that, with respect to~$\kappa$, {\prob{Independent Set}} and {\prob{Clique}} are {\wac}, and this holds even when restricted to regular graphs~\cite{DBLP:journals/cj/Cai08,DBLP:journals/tcs/Marx06,DBLP:conf/cats/MathiesonS08}.

\EPP
{Maximum Partial Vertex Cover (MPVC)}
{A graph~$G=(\vset, \eset)$, and two integers~$\kappa$ and~$t$.}
{Is there $S\subseteq \vset$ such that $\abs{S}\leq \kappa$ and~$S$ covers at least~$t$ edges of~$G$?}

It is known that {\prob{MPVC}} is {\npc} and {\wah} with respect to~$\kappa$~\cite{DBLP:conf/wads/GuoNW05}.

\EPP
{Restricted Exact Cover by Three Sets (RX3C)}
{A universe $\xs=\{\xse_1,\xse_2,\dots,\xse_{3\kappa}\}$ and a collection $\xc=\{\xce_1,\xce_2,\dots,\xce_{3\kappa}\}$ of $3$-subsets of~$\xs$ such that every~$\xse\in \xs$ occurs in exactly three elements of~$\xc$.}
{Does~$\xc$ contain an exact set cover of~$\xs$, i.e., a subcollection $\xc'\subseteq \xc$ such that~$\abs{\xc'}=\kappa$ and every~$\xse\in \xs$ occurs in exactly one element of~$\xc'$?}

It is known that {\prob{RX3C}} is {\npc}~\cite{DBLP:journals/tcs/Gonzalez85}.

\EPP
{Set Packing}
{A universe $\xs$, a collection~$\xc$ of subsets of~$\xs$, and an integer $\kappa\leq \abs{\xs}$.}
{Does~$\xc$ contain a set packing of cardinality $\kappa$, i.e., a subcollection $\xc'\subseteq \xc$ such that $\abs{\xc'}=\kappa$ and for every two distinct $\xce, \xce'\in \xc'$ it holds that $\xce\cap \xce'=\emptyset$?}

It is known that {\prob{Set Packing}} is {\npc} and is {\wac} with respect to $\kappa$. Moreover, the {\wacns} holds even if all elements in~$\xc$ have the same cardinality. This special case is called {\prob{Regular Set Packing}}. However, it should be pointed out that if~$\xc$ consists of~$d$-subsets of~$\xs$ where~$d$ is a constant, the problem becomes {\fpt} with respect to~$\kappa$. See~\cite{DBLP:series/txcs/DowneyF13} for more detailed discussions.

%\EP
%{Vertex Interdiction with Unit Weights and Costs}{VIUWC}
%{A graph $G=(\vset, \eset)$ and two integers~$\kappa$ and~$t$.}
%{Is there a subset $S\subseteq \vset$ of at most~$\kappa$ vertices so that the maximum matching size of the graph $G-S$ is at most~$t$?}
%
%The {\prob{Vertex Matching Interdiction}} problem is a generalization of {\prob{VIUWC}} where each vertex has interdiction cost and each edge has weight, and the question is whether there is a subset~$S$ of vertices of total interdiction weight without exceeding a given threshold value so that after removing~$S$ from $G$ the graph has a maximum weighted matching of weight at most a given number. By a trivial reduction from {\prob{Vertex Cover}},~\citeas{DBLP:journals/dam/Zenklusen10a} showed that {\prob{VIUWC}} is {\nph} even for $t=0$.

In a graph~$G=(\vset, \eset)$, we say that a vertex~$\vere\in \vset$ dominates another vertex $\vere'\in \vset$ if they are neighbors in~$G$, i.e., $\edge{\vere}{\vere'}\in \eset$.

\EPP
{Red-Blue Dominating Set (RBDS)}
{A bipartite graph $G=(R\muplus B, \eset)$ and an integer~$\kappa\leq \abs{B}$.}
{Is there $B'\subseteq B$ such that $\abs{B'}= \kappa$ and~$B'$ dominates~$R$?}

In the definition of {\prob{RBDS}}, we often call vertices in~$R$ red vertices and call those in~$B$ blue vertices.
It is known that {\prob{RBDS}} is {\npc} and, moreover, with respect to~$\kappa$ it is {\wbc}~\cite{fellows2001,garey}.

It is fairly easy to verify that both {\prob{Regular Set Packing}} and {\prob{RBDS}} contain {\prob{RX3C}} as a special case.

\begin{observation}
\label{obs-a}
The following relations  hold. %among the problems {\prob{Regular Set Packing}}, {\prob{RBDS}}, and {\prob{RX3C}}
\begin{itemize}
    \item {\prob{Regular Set Packing}} is equivalent to {\prob{RX3C}} when restricting the former problem so that each $\xce\in \xc$ is of cardinality three and every~$\xse\in \xs$ is contained in three elements of~$\xc$.
    \item {\prob{RBDS}} is equivalent to {\prob{RX3C}} when restricting the {\prob{RBDS}} so that input bipartite graphs are $3$-regular.
\end{itemize}

\end{observation}

%\subsection{Remarks}
%It should be pointed out that our hardness results also hold for corresponding multiwinner voting rules which always select exactly one winning $k$-committee by utilizing a specific tie-breaking scheme (see, e.g.,~\cite{DBLP:journals/aamas/BredereckKN21} for descriptions of some tie-breaking schemes). %In this setting, the question of the control problems is then whether the distinguished candidates can be included in the unique winning $k$-committee (instead of all winning $k$-committees) by performing the corresponding operation.
%In fact, our hardness reductions are established carefully to avoid ties. \onlyfull{For example, the prevalent lexicographic tie-breaking scheme selects the lexicographically smallest tied committees with respect to a fixed order over candidates.
%Other widely studied tie-breaking schemes include the one against the strategic agent and the one in favor of the strategic agent. When several committees are tied, the former one select one including the minimum number of distinguished candidates and the latter one select.}%
%Moreover, our {\poly}- and {\fpt}-algorithms for the additive rules can be adapted to solve these variants when ties are broken lexicographically.

\section{Constructive Control for Sequential Rules}
In this section, we study constructive control for sequential rules.
When $k=1$, sequential $\omega$-Thiele's rules such that $\omega(1)>0$ are equivalent to AV. Such rules are termed {\it{standard}} rules by Lackner and Skowron~\citeas{DBLP:series/sbis/LacknerS23}. It is known that {\probb{CCAV}{AV}} is {\wah} and {\probb{CCDV}{AV}} is {\wbh} with respect to the solution size~$\ell$~\citep{DBLP:journals/ai/HemaspaandraHR07,DBLP:journals/tcs/LiuFZL09}.
It follows that for all standard rules~$\varphi$, {\probb{CCAV}{$\varphi$}} and {\probb{CCDV}{$\varphi$}} are respectively {\wah} and {\wbh} with respect to the solution size even when $k=1$.
However, it is not immediately clear how to modify these reductions for~$k$ being any arbitrary constant.
We derive such a reduction for {\prob{CCAV}} from {\prob{Set Packing}} which is completely different from the reductions established in~\citep{DBLP:journals/ai/HemaspaandraHR07,DBLP:journals/tcs/LiuFZL09}. More importantly, our reduction applies to many sequential $\omega$-Thiele's rules (excluding seqAV), while theirs are tailored for AV (can be adapted to show same results for seqAV).
% the {\wahns} reduction in~\citep{DBLP:journals/tcs/LiuFZL09} is from {\prob{Perfect Code}}, and the {\nphns} reduction in~\citep{DBLP:journals/ai/HemaspaandraHR07} is from {\prob{Exact Cover by Three-Sets}}}.

\begin{theorem}
\label{ccav-seqpav-nph}
For every $\omega$-Thiele's rule such that $\Delta_{\omega}(i)<\omega(1)$ for all positive integers~$i$, {\probb{CCAV}{\memph{seq}$\varphi$}} is {\wah} with respect to~$\ell$, the number of added votes. Moreover, this holds for all positive constants~$k$ and when there is only one distinguished candidate.
\end{theorem}

\begin{proof}
Let $\varphi$ be an $\omega$-Thiele's rule such that $\Delta_{\omega}(i)<\omega(1)$ for all positive integers~$i$. Notice that as~$\omega$ is nondecreasing, this implies that $\omega(1)>0$.
Let $(\xs, \xc, \kappa)$ be an instance of {\prob{Set Packing}}. We construct an instance of {\probb{CCAV}{seq$\varphi$}} as follows. Let~$k=\bigo{1}$ be a positive constant. For every $\xse\in \xs$, we first create one candidate denoted by the same symbol for simplicity, and then we create a set~$C(\xse)$ of~$k$ candidates. Let $C^+(\xse)=C(\xse)\cup \{\xse\}$. In addition, we create a candidate~$p$. Let $J=\{p\}$ and let $C=(\bigcup_{\xse\in \xs}C^+(\xse))\cup J$. Regarding the registered votes, for every $\xse\in \xs$, we create
\begin{itemize}
\item $\kappa-2$ votes approving only $\xse$, and
\item  $\kappa-1$ votes approving only~$c$ for every $c\in C(\xse)$.
\end{itemize}
Let~$V$ be the multiset of the above $\abs{\xs}\cdot (\kappa-2)+\abs{\xs}\cdot k\cdot (\kappa-1)$ registered votes.
The unregistered votes are created according to~$\xc$. In particular, for every $\xce\in \xc$, we create one unregistered vote~$v(\xce)$ which approves~$p$, and all candidates created for the elements in $\xce$, i.e.,
\[v(\xce)=\left(\bigcup_{\xse\in \xce}C^+(\xse)\right) \cup \{p\}.\]
In the following discussion, for a given $\xc'\subseteq \xc$, we use~$U(\xc')$ to denote the set of the unregistered votes corresponding to~$\xc'$, i.e., $U(\xc')=\{v(\xce) \setmid \xce\in \xc'\}$.
Let $\ell=\kappa$. The instance of {\probb{CCAV}{seq$\varphi$}} is $((C, V\muplus U(\xc)), k, J, \ell)$. The tie-breaking order is
$\overrightarrow{\xs} \rhd p\rhd \overrightarrow{\bigcup_{\xse\in \xs}C(\xse)}$.
The instance clearly can be constructed in polynomial time. In the following, we prove the correctness of the reduction.

$(\Rightarrow)$ Assume that there exists $\xc'\subseteq \xc$ of cardinality~$\kappa$ so that every two elements in~$\xc'$ are disjoint.  Let $E=(C, V\muplus U(\xc'))$, and let~$\w$ be the winning $k$-committee of seq$\varphi$ at~$E$ with respect to the tie-breaking order~$\rhd$. By the construction of registered and unregistered votes, in the election~$E$,~$p$ has AV score~$\kappa$, every $\xse\in \xs$ has AV score at most $\kappa-1$, and every other candidate has AV score at most~$\kappa$. Then, as $\omega(1)>0$ and~$p$ is before all the other candidates in~$C\setminus A$ in the tie-breaking order, it holds that~$p$ is the first candidate added into~$\w$ according to the definition of seq$\varphi$.  So, we have $p\in \w$, and hence the instance of {\probb{CCAV}{seq$\varphi$}} is a {\yesins}.

$(\Leftarrow)$ Assume that there exists $U(\xc')\subseteq U(\xc)$ where $\xc'\subseteq \xc$ so that $\abs{\xc'}\leq \ell=\kappa$ and~$p$ is contained in the winning $k$-committee, denoted~$\w$, of seq$\varphi$ at $E=(C, V\muplus U(\xc'))$ with respect to~$\rhd$. We first prove the following claim.

In what follows, for a given $C'\subseteq C$ and a given $\xce\in \xc'$, let $m(C', \xce)=\abs{C'\cap v(\xce)}$ be the number of candidates in~$C'$ approved by~$v(\xce)$.

\begin{claim}
\label{claim-a}
 $\abs{\xc'}=\kappa$.
\end{claim}
{\noindent\it{Proof of Claim~\ref{claim-a}.}}
Assume, for the sake of contradiction, that $\abs{\xc'}<\kappa$. Then, the AV score of~$p$ in~$E$ is at most $\kappa-1$. As $\omega(1)>0$ and~$p$ it not approved by any registered votes,~$p$ is not contained in the winning $k$-committee of seq$\varphi$ at $(C, V)$. As a consequence,~$U(\xc')$ is nonempty. Let~$\xce^{\star}$ be any arbitrary element in~$\xc'$, and let~$\xse^{\star}$ be any arbitrary element in~$\xce$. We compare the~$\varphi$ margin contributions of candidates in~$C(\xse^{\star})$ and that of the distinguished candidate~$p$ to a committee of at most $k-1$ candidates. In particular, let~$C'$ be any arbitrary committee of at most $k-1$ candidates so that $C'\subseteq C\setminus \{p\}$. Let~$c$ be any arbitrary candidate in $C(\xse^{\star})\setminus C'$. As~$\abs{C(\xse^{\star})}=k$ and $\abs{C'}\leq k-1$, such a candidate~$c$ exists. Recall that there are $\kappa-1$ registered votes approving only~$c$. The~$\varphi$ margin contribution of~$c$ to~$C'$ in~$E$ is then
\begin{equation}
\label{eq-xa}
(\kappa-1)\cdot \omega(1)+\sum_{\xse^{\star}\in \xce\in \xc'}\triangle_{\omega}(m(C', \xce)).
\end{equation}
As every unregistered vote approves~$p$, the~$\varphi$ margin contribution of~$p$ to~$C'$ in~$E$ is
\begin{equation}
\label{eq-xb}
    \sum_{\xce\in \xc'}\triangle_{\omega}(m(C', \xce)) = \sum_{\xse^{\star}\not\in \xce\in \xc'}\triangle_{\omega}(m(C', \xce))+ \sum_{\xse^{\star}\in \xce\in \xc'}\triangle_{\omega}(m(C', \xce))\\
\end{equation}
It follows that
\begin{equation*}
\begin{aligned}
     \eqref{eq-xb}-\eqref{eq-xa}= &  \left(\sum_{\xse^{\star}\not\in \xce\in \xc'}\triangle_{\omega}(m(C', \xce))\right)   - (\kappa-1)\cdot \omega(1)  \\
                             \leq & (\kappa-2)\cdot \max_{\xse^{\star}\not\in \xce\in \xc'}\triangle_{\omega}(m(C', \xce)) - (\kappa-1)\cdot \omega(1)\\
                                < & 0. \\
\end{aligned}
\end{equation*}
The second line in the above inequality is due to that $\abs{\xc'}<\kappa$, and~$a^{\star}$ is contained in at least one element of~$\xc'$. The last line is due to that $\triangle_{\omega}(i)<\omega(1)$ for all $i\geq 1$ (and hence $\omega(1)>0$), and the fact that $\triangle_{\omega}(0)=\omega(1)$.
Therefore, the~$\varphi$ margin contribution of~$c$ to~$C'$ is larger than that of~$p$ to~$C'$ in the election~$E$. As~$C'$ can be any committee of $C\setminus \{p\}$ of at most $k-1$ candidates, this means that in the case where $\abs{\xc'}<\kappa$, the existence of the set~$C(\xse^{\star})$ of~$k$ candidates prevents~$p$ from being contained in the winning $k$-committee of seq$\varphi$ at~$E$ with respect to~$\rhd$. In other words, $\abs{\xc'}<\kappa$ contradicts with $p\in \w$. This completes the proof for the claim.
\medskip

By Claim~\ref{claim-a}, to complete the proof of the theorem, it suffices to prove that~$\xc'$ is a set packing. For the sake of contradiction, assume that there exists $\xse\in \xs$ which occurs in at least two elements of~$\xc'$. This implies that the AV score of~$\xse$ in~$E$ is at least~$\kappa$, and the AV score of every candidate in~$C(\xse)$ is at least~$\kappa+1$. As $\abs{\xc'}=\kappa$, the AV score of~$p$ in~$E$ is~$\kappa$.
Therefore, the first candidate added into~$\w$ according to the procedure of seq$\varphi$ must be a candidate from $C\setminus (\{p\}\cup \xs)$ which is approved by at least two votes in~$U(\xc')$.
Without loss of generality, let~$c^{\star}$ denote the first candidate added into~$\w$, and let~$a^{\star}\in \xs$ be such that~$c^{\star}\in C(\xse^{\star})$. Let~$C'$ be a committee such that $c^{\star}\in C' \subseteq C\setminus \{p\}$ and $\abs{C'}\leq k-1$. Let~$t$ be the number of votes in~$U(\xc')$ approving~$c^{\star}$, i.e., $t=\{\xce \setmid \xse^{\star}\in \xce\in \xc'\}$. By the above discussion, it holds that $t\geq 2$. The~$\varphi$ margin contribution of~$p$ to~$C'$ in~$E$ is
\[\sum_{\xce\in \xc'}\Delta_{\omega}(m(C', \xce))\leq (\kappa-t)\cdot \omega(1)+\sum_{\xse^{\star}\in \xce\in \xc'} \Delta_{\omega}(m(C', \xce)).\]
Additionally, the~$\varphi$ margin contribution of every candidate from $C(\xse^{\star})\setminus C'$ to~$C'$ in the election~$E$ is at least
\[(\kappa-1)\cdot \omega(1)+\sum_{\xse^{\star}\in \xce\in \xc'} \Delta_{\omega}(m(C', \xce)),\]
which is larger than that of~$p$ provided $t\geq 2$.
As $\abs{C(\xse^{\star})}=k$, in light of the above discussion,~$p$ cannot be contained in the winning $k$-committee of seq$\varphi$ at~$E$, a contradiction. Therefore, we can conclude now that~$\xc'$ is a set packing of cardinality~$\kappa$.
\end{proof}

A matching in a graph~$G=(\vset, \eset)$ is a subset $\eset'\subseteq \eset$ of edges so that none of two edges in~$\eset'$ share an endpoint in common.
A vertex is saturated by a matching if the vertex is an endpoint of some edge in the matching.

\begin{theorem}
\label{ccdv-seqpav-nph}
For every $\omega$-Thiele's rule such that $\omega(2)<2\omega(1)$, {\probb{CCDV}{\memph{seq}$\varphi$}} is {\nph}.
Moreover, this holds for all positive constants~$k$, and when there is only one distinguished candidate, every vote approves at most three candidates, and every candidate is approved by at most two votes.
\end{theorem}

\begin{proof}
Let $\varphi$ be an $\omega$-Thiele's rule such that $\omega(2)<2\omega(1)$. Observe that by the definition of $\omega$-Thiele's rule, $\omega(2)<2\omega(1)$ implies $\omega(1)>0$.
We prove the theorem by a reduction from {\prob{Vertex Cover}} restricted to $3$-regular graphs. Let~$k=\bigo{1}$ be a positive constant. Let $(G, \kappa)$ be an instance of {\prob{Vertex Cover}} where $G=(\vset, \eset)$ is a $3$-regular graph. We create an instance of {\probb{CCDV}{seq$\varphi$}} as follows. Without loss of generality, we assume that~$G$ contains at least two edges.
Candidates are created as follows.
\begin{itemize}
    \item First, for each edge $\edge{u}{u'}\in \eset$, we create one candidate $c(\edge{u}{u'})$. In the following exposition, for a given $\eset'\subseteq \eset$, we use $C(\eset')=\{c(\edge{u}{u'}) \setmid \edge{u}{u'}\in \eset'\}$ to denote the set of candidates corresponding to~$\eset'$.
    \item In addition, for each $i\in [k-1]$, we create a set $C(i)=\{c_i^1, c_i^2, c_i^3\}$ of three candidates (for $k=1$ we do not have such candidates). Let $C([k-1])=\bigcup_{i\in [k-1]}C(i)$.
    \item Finally, we create one candidate~$p$.
\end{itemize}
Let $J=\{p\}$ and let $C=C(\eset)\cup J\cup C([k-1])$. Obviously, $\abs{C}=\abs{\eset}+3(k-1)+1$.
From $\abs{\eset}\geq 2$, it follows that $\abs{C}\geq 3k>k$.
We create the following votes:
\begin{itemize}
\item two votes approving only the distinguished candidate~$p$;
\item for each vertex $\vere\in \vset$, one vote $v(\vere)$ which approves all the candidates corresponding to edges covered by~$\vere$, i.e., $v(\vere)=\{c(\edge{\vere}{\vere'}) \setmid \vere'\in \neighbor{G}{\vere}\}$;
\item for every $i\in [k-1]$, one set~$V(i)$ of three votes:
\begin{itemize}
\item one vote approving exactly $c_i^1$ and $c_i^2$;
\item one vote approving exactly $c_i^1$ and $c_i^3$; and
\item one vote approving exactly $c_i^2$ and $c_i^3$.
\end{itemize}
\end{itemize}
Let~$V$ denote the multiset of the above constructed votes. Apparently, $\abs{V}=2+\abs{\vset}+3(k-1)=\abs{\vset}+3k-1$. Let $V([k-1])=\bigcup_{i\in [k-1]}V(i)$.
Finally, let $\ell=\kappa$. The instance of {\probb{CCDV}{seq$\varphi$}} is $((C, V), k, J,  \ell)$ which can be constructed in polynomial time. The tie-breaking order is $\overrightarrow{C([k-1])}\rhd \overrightarrow{C(\eset)}\rhd p$. In the following, we prove the correctness of the reduction.

$(\Rightarrow)$ Assume that~$G$ has a vertex cover of~$\kappa$ vertices. Let $V'=\{v(\vere) \setmid \vere\in S\}$ be the set of votes corresponding to~$S$. Let $E=(C, V\setminus V')$. As $G-S$ does not have any edges, in the election~$E$, every edge-candidate~$c(\edge{\vere}{\vere'})\in C(\eset)$ has AV score at most one. Clearly, every candidate in $C([i-1])\cup \{p\}$ has AV score two in~$E$. Then, by (1) the fact that~$p$ is the last one in the tie-breaking order~$\rhd$, (2) the construction of votes for each $i\in [k-1]$, and (3) the fact that $\omega(1)>0$ and  $\omega(2)<2\omega(1)$, we know that the first $k-1$ candidates added into the winning $k$-committees of~seq$\varphi$ at~$E$ are all from $C([k-1])$, one from each~$C(i)$ where $i\in [k-1]$. Formally, it holds that $w^{\leq k-1}_{E}\subseteq C([k-1])$ and, moreover, $\abs{w^{\leq k-1}_{E}\cap C(i)}=1$ holds for all $i\in [k-1]$. Then, by the construction of the election, every candidate in~$C(\eset)$ has~$\varphi$ margin contribution at most~$\omega(1)$ to~$w^{\leq k-1}_{E}$, every candidate in $C([k-1])\setminus w^{\leq k-1}_{E}$ has~$\varphi$ margin contribution~$\Delta_{\omega}(1)+\omega(1)=\omega(2)$ to~$w^{\leq k-1}_{E}$, and~$p$ has~$\varphi$ margin contribution~$2\omega(1)$ to~$w^{\leq k-1}_{E}$. As $\omega(1)>0$ and $\omega(2)<2\omega(1)$, the last candidate added into the winning $k$-committee of~seq$\varphi$ at~$E$ with respect to the tie-breaking order~$\rhd$ must be~$p$, i.e., $w^{k}_{E}=p$.

$(\Leftarrow)$ Assume that there exists $V'\subseteq V$ of at most $\ell=\kappa$ votes so that~$p$ is contained in the winning $k$-committee, denoted~$\w$, of~seq$\varphi$ at the election $E=(C, V\setminus V')$ with respect to~$\rhd$. Observe that~$V'$ contains neither of the two votes approving~$p$. Let
\[S=\{\vere\in \vset \setmid v(\vere)\in V'\}\]
and let
\[t=\abs{V'\cap V([k-1])}.\]
Therefore, $\abs{S}+t=\abs{V'}\leq \kappa$. For each $x\in \{0, 1, 2, 3\}$, let $I_x=\{i\in [k-1] \setmid \abs{V(i)\cap V'}=x\}$, and let $t_x=\abs{I_x}$. It holds that
\begin{equation}
\label{eq-ccc}
    t=\sum_{x\in [3]} x \cdot t_x.
\end{equation}
Without loss of generality, assume that~$p$ is the $j$-th candidate added into~$\w$ for some integer $j\in [k]$, i.e., $\w^j_{E}=p$. By the definition of~$\rhd$, the fact that $\omega(1)>0$ and $\omega(2)<\omega(1)$, and that the two votes approving~$p$ is excluded from~$V'$, we know that for every $j'<j$  the~$\varphi$ margin contribution of~$\w^{j'}_E$ to~$\w^{\leq j'-1}_E$ is at least $2\omega(1)$. Let $C(\eset')=\w^{\leq j-1}_{E}\cap C(\eset)$ where $\eset'\subseteq \eset$. The above discussion implies that~$A'$ forms a matching of $G-S$. More importantly, the matching~$A'$ is maximal in $G-S$, since otherwise there exists a candidate in $C(A)\setminus \w^{\leq j-1}_E$ whose~$\varphi$ margin contribution to~$\w^{\leq j-1}_E$ is~$2\omega(1)$ and then by the tie-breaking order~$p$ cannot be~$\w^{j-1}_E$, a contradiction. Let~$S'$ be the set of vertices saturated by~$A'$. So, we have that
\begin{equation}
\label{eq-aa}
    \abs{S'}=2\abs{A'}.
\end{equation} Obviously, $S\cup S'$ is a vertex cover of~$G$. To complete the proof, it suffices to show that $\abs{S}+\abs{S'}\leq \kappa$. To this end, let $B=\w^{\leq j-1}_E\cap C([k-1])$. Analogous to the above reasoning, we know that for each $c\in B$ which is the $j'$-th candidate added into~$\w$ for some $j'\leq j$, the~$\varphi$ margin contribution of~$c$ to $\w^{\leq j'-1}_E$ is at least $2\omega(1)$. This implies that for each $c\in B\cap C(i)$ where $i\in [k-1]$, neither of the two votes in~$V(i)$ approving~$c$ is contained in~$V'$. Moreover, for each $i\in I_0\cup I_1$, there exists at least one candidate in~$C(i)$ which is approved by at least two votes in $V(i)\setminus V'$. So, the~$\varphi$ margin contribution of this candidate to every committee disjoint from~$C(i)$ is $2\omega(1)$. As~$p$ is the last candidate in the tie-breaking order, it holds that for every $i\in I_0\cup I_1$, $\w^{\leq j-1}_E$ contains one candidate from~$C(i)$. Summarizing these discussions, we obtain
$\abs{C(A')}+\abs{B}=\abs{A'}+(k-1-t_2-t_3)$. Then, from $C(A')\cup B=\w^{\leq j-1}_E$ and $\abs{\w^{\leq j-1}_E}=j-1\leq k-1$, we obtain
\begin{equation}
\label{eq-bb}
    \abs{A'}\leq t_2+t_3\leq \frac{t}{2}\leq \frac{\kappa-\abs{S}}{2},
\end{equation}
where $t_2+t_3\leq \frac{t}{2}$ is by~\eqref{eq-ccc}.
Putting~\eqref{eq-aa} and \eqref{eq-bb} together yields $\abs{S}+\abs{S'}\leq \kappa$. The proof is completed.
\end{proof}

Now we study constructive control by adding candidates. As single winner AV is immune to CCAC~\cite{DBLP:journals/ai/HemaspaandraHR07}, when $k=1$ sequential $\omega$-Thiele's rules are immune to CCAC (assuming ties broken by a predefined linear order). However, in general this is not the case, as illustrated in Example~\ref{ex-a}.

\begin{example}[sequential $\omega$-Thiele's rules are susceptible to CCAC]
\label{ex-a}
Consider an election~$E$ with the following four candidates and three votes.
\begin{center}
\begin{tabular}{c cccc}\hline
       & $a$       & $b$         &~$p$            & $c$ \\ \hline

$v_1$  & $\checkmark$  &             &                & $\checkmark$ \\

$v_2$  & &     $\checkmark$        &                & $\checkmark$ \\

$v_3$  &  &             & $\checkmark$        & \\ \hline
\end{tabular}
\end{center}
Here, a check mark means that the corresponding voter approves the corresponding candidate.
Let the tie-breaking order be $a\rhd b\rhd p\rhd c$. The winning $2$-committee of seqPAV at the election~$E$ restricted to $\{a, b, p\}$ is $\{a, b\}$. However, the winning $2$-committee of seqPAV at~$E$ is $\{c, p\}$.
\end{example}

In fact, Example~\ref{ex-a} shows that for all $\omega$-Thiele's rules~$\varphi$ such that $\omega(2)<2\omega(1)$, seq$\varphi$ is susceptible to CCAC even when $k=2$. The above example can be easily extended to show the susceptibility of all  $\omega$-Thiele's rules~$\varphi$ such that $\omega(2)<2\omega(1)$ to {\prob{CCAC}} for all $k\geq 2$.

Regarding the complexity of {\probb{CCAC}{seq$\varphi$}}, we have the following theorem.

\begin{theorem}
\label{thm-ccac-seqThiele-wbh}
For every $\omega$-Thiele's rule~$\varphi$ such that $\Delta_{\omega}(i)<\omega(1)$ for all positive integers~$i$, {\probb{CCAC}{\memph{seq}$\varphi$}} is {\wbh} with respect to~$k+\ell$, the size of a winning committee plus the number of added candidates. Moreover, this holds even if there is only one distinguished candidate.
\end{theorem}

\begin{proof}
Let~$\varphi$ be an $\omega$-Thiele's rule~$\varphi$ such that $\Delta_{\omega}(i)<\omega(1)$ for all positive integers~$i$.
We prove the theorem via a reduction from {\prob{RBDS}}. Let $(G, \kappa)$ be an instance of {\prob{RBDS}}, where~$G$ is a bipartite graph with the bipartition $(R, B)$. Without loss of generality, we assume that~$G$ does not contain any isolated vertices, and assume that $\abs{R}\geq \kappa+1$. These assumptions do not change the {\wbhns} of {\prob{RBDS}}.\footnote{The assumption that~$G$ does not contain any isolated vertices is clear. Then, if $\abs{R}\leq \kappa$, the instance is a trivial {\yesins} because by selecting, for each red vertex any arbitrary one blue vertex dominating the red vertex, we obtain a subset of at most~$\kappa$ blue vertices which dominate~$R$. Apparently, any $\kappa$-set of blue vertices containing this subset also dominates~$R$.}
We construct an instance of {\probb{CCAC}{seq$\varphi$}} as follows.
For each $\vere\in R\cup B$, we create one candidate~$c(\vere)$. Let $C(R)=\{c(r) \setmid r\in R\}$. In the following analysis, for a given $B'\subseteq B$, we use~$C(B')$ to denote the set of candidates corresponding to~$B'$, i.e., $C(B')=\{c(b) \setmid b\in B'\}$.
In addition, we create one candidate~$p$. Let $C=\{p\}\cup C(R)$, $D=C(B)$, and $J=\{p\}$. We create the following votes:
\begin{itemize}
\item one vote approving only the candidate~$p$;
\item for each $r\in R$, one vote $v(r)=\{c(r)\}\cup C(\neighbor{G}{r})$ approving the candidate corresponding to~$r$ and all candidates corresponding to the neighbors of~$r$ in~$G$.
\end{itemize}
Let $V$ be the multiset of the above created $1+\abs{R}$ votes. The tie-breaking order is $\overrightarrow{D}\rhd \overrightarrow{C}\rhd p$. Finally, let $\ell=\kappa$ and let $k=\kappa+1$. The instance of {\probb{CCAC}{seq$\varphi$}} is $((C\cup D, V), k, J, \ell)$ which can be constructed in polynomial time.
In the following, we prove the correctness of the reduction.

$(\Rightarrow)$ Assume that there exists $B'\subseteq B$ of~$\kappa$ vertices which dominate~$R$.
Let $E=(C\cup C(B'), V)$.  Recall that~$G$ does not contain any isolated vertices. Then, by the tie-breaking order~$\rhd$ and the fact that every~$c(r)$ where~$r\in R$ is approved by only one vote in~$V$, we know that the first candidate added into~$\w$ according to seq$\varphi$ must be from~$C(B')$, i.e., $\w^1_{E}\in C(B')$. Now, as $\abs{C(B')}=\kappa=k-1$, to prove that~$p$ is contained in the winning $k$-committee of seq$\varphi$ at~$E$, it suffices to prove that for every $i\in [k-1]$, if $\w^{\leq i}_{E}\subseteq C(B')$, then $\w^{i+1}_{E}\in C(B')\cup \{p\}$. In other words,  before any red-candidate could enter the winning $k$-committee of seq$\varphi$ at~$E$, the distinguished candidate~$p$ is already in the winning committee.
To this end, let us assume that $\w^{\leq i}_{E}\subseteq C(B')$ for some $i\in [k-1]$. Additionally, let $B'(i)=\{b\in B' \setmid c(b)\in \w^{\leq i}_E\}$ be the set of vertices in~$B'$ whose corresponding candidates are contained in~$\w^{\leq i}_E$.
For each $r\in R$, we use~$x_r$ to denote the number of vertices in~$B'(i)$ dominating~$r$, i.e., $x_r=\abs{\neighbor{G}{r}\cap B'(i)}$.
Let $r^{\star}\in R$ be any arbitrary candidate in~$R$. The~$\varphi$ margin contribution of~$r^{\star}$ to~$\w^{\leq i}_{E}$ is then~$\Delta_{\omega}(x_{r^{\star}})$. We prove below that there exists at least one candidate from $C(B')\setminus \w^{\leq i}_{E}$ whose~$\varphi$ margin contribution to~$\w^{\leq i}_{E}$ is no smaller than that of~$r^{\star}$, or~$p$ has a strictly larger margin contribution to~$\w^{\leq i}_{E}$ than that of~$r^{\star}$. Our proof is done by distinguishing between the following two cases.
\begin{itemize}
\item {\bf{Case~1:}} $r^{\star}$ is dominated by some vertex in~$B'(i)$

In this case,  $x_{r^{\star}}\geq 1$.
Then, by the stipulation of the theorem, we know that $\Delta_{\omega}(x_{r^{\star}})<\omega(1)$.
However, the~$\varphi$ margin contribution of~$p$ to~$\w^{\leq i}_{E}$ is~$\omega(1)$.

\item {\bf{Case~2:}}  $r^{\star}$ is not dominated by any vertex in~$B'(i)$

As~$B'$ dominates~$R$, there exists $b\in B'\setminus B'(i)$ which dominates~$r^{\star}$. Then, the~$\varphi$ margin contribution of~$c(b)$ to~$\w^{\leq i}_{E}$ is $\sum_{r\in R, b\in \neighbor{G}{r}}\Delta_{\omega}(x_r)$, which is no smaller than that of~$r^{\star}$ given $b\in \neighbor{G}{r^{\star}}$.
\end{itemize}
Now we know that $\w^{i+1}_{E}\in C(B')\cup \{p\}$. As discussed above, this implies that~$p$ is contained in the winning $k$-committee of seq$\varphi$ at~$E$.

$(\Leftarrow)$ Suppose that there exists $B'\subseteq B$ so that $\abs{B'}\leq \ell=\kappa$ and~$p$ is contained in the winning $k$-committee of~seq$\varphi$, denoted~$w$, at $E=(C\cup C(B'), V)$. We construct a subset $\widetilde{B}\subseteq B$ of at most~$\kappa$ vertices which dominate~$R$ as follows. Let~$B'_w=C(B')\cap w$ be the set of blue vertices corresponding to blue-candidates contained in~$w$. Let~$R'$ be the set of red vertices dominated by~$B'_w$, i.e., $R'=\neighbor{G}{B'_w}$.
%Let~$C'\subseteq (C\cup C(B')\setminus \{p\})$ be a committee of at most $k-1$ candidates excluding~$p$.
%Let \[T=\{c(x) \setmid x\in R\cup B', (\exists y\in R\cup B')[c(y)\in C', x\in \neighbor{G}{y}]\}\] be the set of candidates from~$C(R\cup B')$ whose corresponding vertices are dominated by vertices corresponding to candidates in~$C'$.
\begin{observation}
\label{ob-1}
Let~$C'$ be a committee in~$E$ so that $p\not\in C'$ and $C'\setminus C(R\setminus R')\neq \emptyset$.
  Then, the~$\varphi$ margin contribution of every candidate from $C(R\setminus R')\setminus C'$ to~$C'$, and that of~$p$ to~$C'$ are~$\omega(1)$.
\end{observation}
%\begin{observation}
%\label{obs-2}
%    The~$\varphi$ margin contribution of every candidate from $C(B')\setminus (C'\cup T)$ to~$C'$ is at least~$\omega(1)$.
%\end{observation}
%By Observations~\ref{ob-1}--~\ref{obs-2} and the definition of the tie-breaking order, we infer that~$w$ is formed in the following order:
%\begin{itemize}
%\item Initialized with~$C'=\emptyset$, candidates in a subset of candidates from~$C(B')$ are added into~$C'$ one after another according to their margin contribution to the committee~$C'$ at that time, until the blue vertices corresponding to~$C'$ dominate~$R'$. Let~$B''$ be the set of blue vertices corresponding to~$C'$ at this point in time.
%\item Then,~$C'$ is further expanded according to the definition of seq$\varphi$ until it contains exactly~$k$ candidates. At this point in time it holds that $w=C'$.
%\end{itemize}
By Observation~\ref{ob-1} and the definition of the tie-breaking order, we know that before~$p$ is added into the winning $k$-committee~$w$, all candidates from $C(R\setminus R')$ are already added into~$w$. It follows that $\abs{B'_w}+\abs{R\setminus R'}\leq k-1=\kappa$. As~$G$ does not contain any isolated vertices, for any $r\in R\setminus R'$, we fix an arbitrary neighbor of~$r$ in~$B$, and use~$b(r)$ to denote this vertex. Let $\widetilde{B}=B'_w\cup \{b(r) \setmid r\in R\setminus R'\}$. Clearly, $\abs{\widetilde{B}}\leq \abs{B'_w}+\abs{R\setminus R'}$, and~$\widetilde{B}$ dominates~$R$. It follows that the {\prob{RBDS}} instance is a {\yesins}.
\end{proof}

%As {\prob{RBDS}} remains {\nph} when every blue vertex has degree at most three and every red vertex has degree two, the proof of Theorem~\ref{thm-ccac-seqThiele-wbh} implies the following corollary.

%\begin{corollary}
%\label{cor-ccac-thiele-np-hard}
%For every $\omega$-Thiele's rule $\varphi$ such that $\Delta_{\omega}(i)<\omega(1)$ for all $i\in [3]$,
%{\probb{CCAC}{\memph{seq}$\varphi$}} is {\nph} even when there is only one distinguished candidate, every voter approves at most three candidates, and every candidate is approved by at most three voters.
%\end{corollary}

%\begin{theorem}
%\label{thm-ccac-seqabccv-np-hard}
%{\prob{CCAC}{\memph{seqABCCV}}} is {\nph}. Moreover, this holds even if every voter approves at most four candidates, every candidate is approved by at most three voters, and there is only one distinguished candidate.
%\end{theorem}
%
%\begin{proof}
%We prove the theorem by adding some additional voters in the reduction for {\probb{CCAC}{seqPAV}}. In particular, for every $\xce\in\xc$, we create one more vote which approves only the candidate $c(\xce)$. This is to ensure that once a candidate from $D$ is added into $C$, this candidate will be added into the winning committee.
%\end{proof}

%\onlyfull
{For the operation of deleting candidates, we have the following hardness result.}

\begin{theorem}
\label{thm-ccdc-seqpav-seqabccv-np-hard}
For every $\omega$-Thiele's rule $\varphi$ such that $\Delta_{\omega}(i)<\omega(1)$ for all $i\in [3]$, {\probb{CCDC}{\memph{seq}$\varphi$}} is {\nph}. Moreover, this holds even if there is only one distinguished candidate, every vote approves at most four candidates, and every candidate is approved by at most three votes.
\end{theorem}

\begin{proof}
We prove the theorem by a reduction from {\prob{RX3C}}. The reduction is similar to the one in the proof of Theorem~\ref{thm-ccac-seqThiele-wbh}. Note that {\prob{RX3C}} is a special case of {\prob{RBDS}}. In particular, Let $(\xs, \xc)$ be an instance of {\prob{RX3C}}, where $\abs{\xs}=\abs{\xc}=3\kappa$ for some positive integer~$\kappa$. Then, if we consider~$\xs$ as red vertices and consider~$\xc$ as blue vertices so that the neighbors of each $\xce\in \xc$ are exactly the three elements in~$\xce$, we obtain a special instance of {\prob{RBDS}}.

We create an instance of {\probb{CCDC}{seq$\varphi$}} as follows. For each  $x\in \xs\cup \xc$, we create one candidate~$c(x)$. In the following analysis, for a given $\xc'\subseteq \xc$ (resp.~$\xs'\subseteq \xs$), we use~$C(\xc')$ (resp.\ $C(\xs')$) to denote the set of candidates corresponding to~$\xc'$ (resp.~$\xs'$), i.e., $C(\xc')=\{c(\xce) \setmid \xce\in \xc'\}$ (resp.\ $C(\xs')=\{c(\xse) \setmid \xse\in \xs'\}$). In addition, we create one candidate~$p$. Let $J=\{p\}$ and let $C=C(\xs)\cup C(\xc)\cup J$. We create the following votes:
\begin{itemize}
    \item one vote approving only the candidate~$p$;
    \item for each $\xse \in \xs$, one vote $v(\xse)=\{c(\xse)\}\cup \{c(\xce)\setmid \xse\in \xce\in \xc\}$; and
    \item for each $\xce\in \xc$, one vote $v(\xce)=\{c(\xce)\}$ approving only~$c(\xce)$.
\end{itemize}
Let $k=\kappa+1$ and $\ell=2\kappa$. The tie-breaking order is $\overrightarrow{C(\xc)}\rhd \overrightarrow{C(\xs)}\rhd p$. Observe that the election $(C, V)$ is in fact obtained from the one constructed in the way described in the proof of Theorem~\ref{thm-ccac-seqThiele-wbh} (based on the special {\prob{RBDS}} instance corresponding to $(\xs, \xc)$) by adding~$\abs{\xc}$ new votes, one for each $\xce\in \xc$ that approves only $c(\xce)$. These votes are to ensure that to make~$p$ be contained in the winning $k$-committee by deleting at most~$\ell$ candidates, exactly $\ell=2\kappa$ candidates from~$C(\xc)$ need to be deleted.
 The instance of {\probb{CCDC}{seq$\varphi$}} is $((C, V), k, J, \ell)$ which clearly can be constructed in polynomial time.
It remains to prove the correctness.

$(\Rightarrow)$ By following the argument for the~$(\Rightarrow)$ direction in the proof of Theorem~\ref{thm-ccac-seqThiele-wbh}, one sees that if~$\xc$ contains an exact set cover~$\xc'$ of cardinality~$\kappa$, then~$p$ is contained in the winning $k$-committee of seq$\varphi$ at $E=(C(A)\cup C(\xc')\cup \{p\}, V)$. More precisely, the winning $k$-committee of seq$\varphi$ at~$E$ is exactly $C(\xc')\cup \{p\}$, with first the candidates in~$C(\xc')$ being added into the winning committee one-by-one according to the tie-breaking order, and then~$p$ being the last added into the winning committee. To see that the first~$\kappa$ candidates added into the winning committee are exactly those from~$C(\xc')$, observe first that the first candidate added into the winning committee is from~$C(\xc')$. Then, assuming $\w^{\leq i}_E\subseteq C(\xc')$ for some $i\leq \kappa-1$, we show that $\w^{i+1}_E\in C(\xc')$. To this end, observe just that the~$\varphi$ margin contribution of each candidate~$c(\xse)$ where~$\xse\in \xs$ to $\w^{\leq i-1}_E$ is at most~$\max_{j\in [3]\cup \{0\}} \triangle_{\omega(j)}$, which is no greater than~$\omega(1)$. However, the~$\varphi$ margin contribution of every candidate in $C(\xc')\setminus w^{\leq i-1}_E$ to $\w^{\leq i-1}_E$ is at least~$\omega(1)$, and that of~$p$ is $\omega(1)$. By the tie-breaking order, it holds $\w^{i+1}_E\in C(\xc')\setminus \w^{\leq i-1}_E$. To see that~$p$ is the last candidate added into the winning committee, recall that~$\xc'$ is an exact set cover of~$\xs$. As a consequence, every~$c(\xse)$ where $\xse\in \xs$ has~$\varphi$ margin contribution to~$C(\xc')$ at most $\triangle_{\omega}(1)$. This is smaller than that of~$p$ which is~$\omega(1)$.

$(\Leftarrow)$ %It remains to prove the correctness for the opposite direction.
Assume that there exists $C'\subseteq C\setminus \{p\}$ so that $\abs{C'}\leq \ell=2\kappa$ and~$p$ is contained in the winning $k$-committee, denoted~$\w$, of~seq$\varphi$ at $E=(C\setminus C', V)$. By the definition of the tie-breaking order and the votes created for~$\xc$, we know that $C'\subseteq C(\xc)$, and moreover $\abs{C'}=\ell=2\kappa$. The reason is that, if for the sake of contradiction this is not the case, then there will remain at least $\kappa+1$ candidates in $C(\xc)\setminus C'$. As every candidate $c(\xce')\in C(\xc)\setminus C'$ is approved by a unique vote~$v(\xce)=\{c(\xce)\}$, the~$\varphi$ margin contribution of~$c(\xce)$ to any committee without $c(\xce)$ and~$p$ is no smaller than that of $p$ to this committee. Therefore, $p\in \w$ implies $(C(\xc)\setminus C')\subseteq \w$. However, this contradicts $\abs{\w}=k=\kappa+1$ and $\abs{C(\xc)\setminus C'}\geq \kappa+1$.
By replacing~$p$ with~$c(\xse)$ for $\xse\in \xs$ in the above argument, we know that a candidate from~$C(\xs)$ can be contained in~$\w$ only if all candidates in $C(\xc)\setminus C'$ are contained in~$\w$. This implies that the first~$\kappa$ candidates added into~$\w$ are exactly those from~$C(\xc)\setminus C'$, i.e., $C(\xc)\setminus C'=\w\setminus \{p\}$.
%discussion in the proof of Theorem~\ref{thm-ccac-seqThiele-wbh} ($(\Leftarrow)$ direction), winning candidates are added into~$\w$ in the following order:

%As a matter of fact, it is fairly easy to see that winning candidates are added into~$\w$ in the following order:
%\begin{enumerate}
%   \item[(1)] First, candidates in $C(\xc)\setminus C'$ are added into~$\w$ according to their $\varphi$ margin contributions until the $3$-subsets corresponding to candidates added in~$\w$ cover~$\xs$.
 %   \item[(2)] Then candidates $c(\xse)\in C(\xs)\setminus C'$ so that~$\xse$ is not covered by~$\xc'$ are one-to-one added into~$\w$. Let $\xs'=\{\xse\in \xs \setmid c(\xse)\in C(\xs)\setminus C', \xse\not\in %\bigcup_{\xce\in \xc'}\xce\}$ be the elements in $\xs$ whose corresponding candidates are added into~$\w$.
%    \item[(2)] Then,~$p$ is added into~$\w$.
%\end{enumerate}
Let~$\xc'=\{\xce\in \xc \setmid c(\xce)\in C(\xc)\setminus C'\}$. We claim that~$\xc'$ must cover~$\xs$. In fact, if this is not the case, there exists~$c(\xse)\in C(\xs)$ not covered by~$\xc'$. Then, the~$\varphi$ margin contribution of $c(\xse)$ to $\w\setminus \{p\}$ is~$\omega(1)$, equaling that of~$p$ to $\w\setminus \{p\}$. This implies that~$p$ cannot be the last candidate added into~$\w$, a contradiction. As $\abs{\xc'}=\abs{C(\xc)\setminus C'}=\abs{C(\xc)}-2\kappa=\kappa$ and $\xc'$ covers $\xs$, the instance of {\prob{RX3C}} is a {\yesins}.
%However, as $k=\kappa+1$ and $\abs{C(\xc')}=\abs{\xc'}=\kappa$, if $p\in \w$ it must hold that $\xs'=\emptyset$, which indicates that~$\xc'$ is an exact set cover of~$\xs$.
%From this, we observe that~$C'$ in fact does not contain any candidate from~$C(\xs)$, i.e., $C(\xs')=\emptyset$, since otherwise it would hold $\abs{C(\xc')}+\abs{C(\xs')}=\abs{(\xc')}+\abs{(\xs')}\geq k=\kappa+1$ which contradicts with $p\in \w$. In fact, if~$C'$ contained~$i$ candidates from~$C(\xs)$ for some positive integer~$i$. Then, we have $\abs{\xc'}\leq \kappa-i$. As each $\xce\in \xc'$ is a $3$-set, $\abs{\xc'}$ can cover at most $3\abs{\xc'}$ elements of~$\xs$. This implies that $\abs{\xs'}\geq 3\kappa-i-(3\kappa-3i)=2i$, and hence $\abs{\xc'}+\abs{\xs'}\geq 3\kappa-2\abs{\xc'}+2i$, which is strictly larger than~$\kappa$ because $\abs{\xc'}\leq \kappa$ and $i\geq 1$.
%
%From this observation, we may assume that $C'=C(\xc')$. Then, by~(1),~$\xc'$ is a set cover of $\xs$. As $\abs{\xc'}\leq \kappa$, $\xc'$ is an exact set cover of $\xs$.
%\begin{claim}
%%\label{claim-h}
%$C'\cap C(\xs)=\emptyset$.
%\end{claim}
%
%{\noindent{\textit{Proof of Claim~\ref{claim-h}}}}.
%The reduction for {\probb{CCDC}{{seqPAV}}} ({\probb{CCDC}{seqABCCV}}) is similar to the one for {\probb{CCAC}{seqPAv}} (resp.\ {\prob{CCAC}{seqABCCV}}) presented in the proof of Theorem~\ref{thm-ccac-seqPav-np-hard} (resp.\ Theorem~\ref{thm-ccac-seqabccv-np-hard}), with the only difference that all candidates are registered candidates, and we are allowed to delete at most $2\kappa$ candidates, i.e., $\ell=2\kappa$.
\end{proof}

{Note that Theorems~\ref{ccav-seqpav-nph}--\ref{thm-ccdc-seqpav-seqabccv-np-hard} apply to seqABCCV and seqPAV. }

\section{Destructive Control}

In this section, we investigate the complexity of destructive control by adding/deleting votes/candidates.
For ABCCV, PAV, and MAV, we show that even checking if a particular candidate is not contained in any winning $k$-committees is computationally hard.

\begin{theorem}
\label{thm-abccv-very-speical-case-np-hard}
{\prob{Nonwinning Testing}} for {\memph{ABCCV}} and {\prob{Nonwinning Testing}} for {\memph{PAV}} are {\wah} with respect to~$k$. Moreover, this holds even when every vote approves at most two candidates.
\end{theorem}

\begin{proof}
    We prove the theorem by reductions from {\prob{Independent Set}} restricted to regular graphs. Let $(G, \kappa)$ be an instance of {\prob{Independent Set}} where~$G$ is a regular graph. Let~$d$ be the degree of vertices of~$G$. We show the hardness for ABCCV and PAV separately as follows.

    \begin{itemize}
        \item ABCCV

        For each vertex of~$G$ we create one candidate denoted by the same symbol for simplicity. In addition, we create one candidate~$p$. Let $C=\vvset{G}\cup \{p\}$, and let $J=\{p\}$. We create the following votes:
        \begin{itemize}
            \item $d-1$ votes each approving exactly the distinguished candidate~$p$; and
            \item for every edge $\edee=\edge{\vere}{\vere'}$ in~$G$,  one vote $v(\edee)=\{\vere,\vere'\}$.
        \end{itemize}
        Let~$V$ denote the multiset of the above created votes. Clearly, every vote approves at most two candidates, and every candidate is approved by at most~$d$ votes. Finally, let $k=\kappa$. The instance of {\prob{Nonwinning Testing}} for ABCCV is $((C, V), k, p)$ which can be constructed in polynomial time. We prove the correctness of the reduction as follows.

        $(\Rightarrow)$ Assume that~$G$ contains an independent set~$w$ of size~$\kappa$. By the construction of the votes, every $c\in w$ is approved by exactly~$d$ votes which do not approve any other candidates from~$w$. So, the committee~$w$ has the maximum possible ABCCV score $\kappa\cdot d$. As~$p$ is approved by exactly $d-1$ votes which approve only~$p$, we know that~$p$ cannot be in any winning $k$-committee of~$(C, V)$.

        $(\Leftarrow)$ Assume that~$p$ is not contained in any winning $k$-committee of~$(C, V)$. Without loss of generality, let~$\vset'\subseteq \vset$ be a winning $k$-committee of~$(C, V)$. We show that~$\vset'$ is an independent set of~$G$. For the sake of contradiction, assume that this is not the case. Then, there exist $\vere, \vere'\in \vset'$ so that $\vere$ and $\vere'$ are adjacent in~$G$. By the construction of the votes, the vote corresponding to the edge between $\vere$ and $\vere'$ approve both $\vere$ and $\vere'$. This implies that the committee $\vset'\setminus \{\vere\}\cup \{p\}$ has at least the same ABCCV score as $\vset'$. However, this contradicts that $p$ is not contained in any winning $k$-committee of~$(C, V)$. So, we can conclude that $\vset'$ is an independent set of~$G$, and hence the {\prob{Independent Set}} instance is a {\yesins}.
    \end{itemize}

    \begin{itemize}
        \item PAV

        The set $C$ and $J$ are defined the same as in the above reduction for ABCCV. We create the following votes:
        \begin{itemize}
            \item $2d-1$ votes each approving exactly the distinguished candidate~$p$; and
            \item for every edge $\edee=\edge{\vere}{\vere'}$ in~$G$, two votes $v(\edee)$ and $v'(\edee)$ each approving exactly~$\vere$ and~$\vere'$.
        \end{itemize}
        Let~$V$ denote the multiset of the above created votes. Clearly, every vote approves at most two candidates, and every candidate is approved by at most~$2d$ votes. Finally, let $k=\kappa$. The instance of {\prob{Nonwinning Testing}} for ABCCV is $((C, V), k, p)$ which can be constructed in polynomial time. The correctness of the reduction is similar to the arguments in the above proof for ABCCV.
%
 %       $(\Rightarrow)$ Assume that~$G$ contains an independent set~$w$ of size~$\kappa$. By the construction of the votes, every candidate $c\in w$ is approved by exactly~$d$ votes which do not approve any other candidates from~$w$. So, the committee~$w$ has the maximum possible ABCCV score $2d\cdot \kappa$. As~$p$ is approved by exactly $2d-1$ votes which approve only~$p$, we know that~$p$ cannot be in any winning $k$-committee of~$(C, V)$.
%
 %       $(\Leftarrow)$ Assume that~$p$ is not contained in any winning $k$-committee of~$(C, V)$. Without loss of generality, let~$\vset'\subseteq \vset$ be a winning $k$-committee of~$(C, V)$. We show that~$\vset'$ is an independent set of~$G$. For the sake of contradiction, assume that this is not the case. Then, there exist $\vere, \vere'\in \vset'$ so that $\vere$ and $\vere'$ are adjacent in~$G$. Then, by the construction of the votes, the vote corresponding to the edge between $\vere$ and $\vere'$ approve both $\vere$ and $\vere'$. This implies that the committee $\vset'\setminus \{\vere\}\cup \{p\}$ has at least the same ABCCV score as $\vset'$. However, this contradicts that $p$ is not contained in any winning $k$-committee of~$(C, V)$. So, we can conclude that $\vset'$ is an independent set of~$G$, and hence the {\prob{Independent Set}} instance is a {\yesins}.
    \end{itemize}
\end{proof}

For MAV, the problem is more computationally demanding.

\begin{theorem}
\label{thm-mav-very-speical-case-np-hard}
{\prob{Nonwinning Testing}} for {\memph{MAV}} is {\wbh} with respect to~$k$.
\end{theorem}

\begin{proof}
    We prove the theorem via a reduction from {\prob{RBDS}}. Let $(G, \kappa)$ be an instance of {\prob{RBDS}} where~$G$ is a bipartite graph with the bipartition $(R, B)$. We assume that~$G$ does not contain any isolated vertices, and all red vertices have the same degree~$d$. This assumption does not change the {\wbhns} of the problem.\footnote{The assumption that~$G$ contains no isolated vertices is clear. If red vertices do not have the same degree, let~$d$ be the maximum degree of red vertices in~$G$. Then, for each $r\in R$ whose degree is at most $d-1$, we create $d-\degree{G}{r}$ blue vertices each of which is only adjacent to~$r$. It is easy to verify that the original instance and the new instance are equivalent.}. We create the following candidates.
\begin{itemize}
    \item For each blue vertex we create one candidate denoted by the same symbol for simplicity.
    \item For each red vertex~$r$, we create a set $C(r)$ of~$\abs{B}+2$ candidates. Let $C(R)=\bigcup_{r\in R}C(r)$ be the set of the $\abs{R}\cdot (\abs{B}+2)$ candidates created for all red vertices.
    \item In addition, we create a candidate~$p$.
    \item Finally, we create a set~$X$ of~$2\kappa+d$ candidates.
\end{itemize}
Let $C=B\cup C(R)\cup X\cup \{p\}$ be the set of the above created candidates. We create the following votes:
    \begin{itemize}
        \item one vote~$v$ which approves exactly the $\abs{B}+2\kappa+d$ candidates in $B\cup X$;
        \item for every red vertex $r\in R$, one vote $v(r)=\neighbor{G}{r}\cup C(r)$ approving exactly the neighbors of~$r$ in~$G$ and the $\abs{B}+2$ candidates created for $r$.
    \end{itemize}
    Let $V$ denote the multiset of the above created $1+\abs{R}$ votes.
    Let $k=\kappa$. The instance of {\prob{Nonwinning Testing}} for {MAV} is~$((C, V), k, p)$ which can be constructed in polynomial time. We prove the correctness of the reduction as follows.

    $(\Rightarrow)$ Assume that there is a subset $B'\subseteq B$ such that $\abs{B'}=\kappa$ and~$B'$ dominates~$R$. By the construction of votes, the Hamming distance between $B'$ and every vote $v(r)$ is at most $(d-1+\kappa-1)+(\abs{B}+2)=\abs{B}+d+\kappa$, and the Hamming distance between $B'$ and the vote $v$ is exactly $\abs{B}+2\kappa+d-\kappa=\abs{B}+d+\kappa$, which is the minimum possible Hamming distance between a $k$-committee and the vote $v$, since $v$ contains all candidates from~$B'$. Then, as~$p$ is not approved by~$v$, the Hamming distance between any $k$-committee containing~$p$ and the vote $v$ is strictly larger than $\abs{B}+d+\kappa$. Therefore,~$p$ cannot be in any winning $k$-committee of $(C, V)$.

    $(\Leftarrow)$ Assume that $p$ is not contained in any winning $k$-committee of $(C, V)$. Let~$w$ be a winning $k$-committee of $(C, V)$, and let $B'=B\cap w$. We show below that~$B'$ dominates~$R$. Assume, for the sake of contradiction, that this is not the case. Let~$r\in R$ be a red vertex not dominated by~$B'$. Note that the MAV score of any $k$-committee contained in~$B$ is at most $\abs{B}+d+\kappa+2$. Therefore, the MAV score of~$w$ is at most $\abs{B}+d+\kappa+2$. However, if the MAV score of~$w$ is exactly $\abs{B}+d+\kappa+2$, it is easy to verify that any $k$-committee containing $k-1$ candidates from~$B$ and~$p$ also has MAV score $\abs{B}+d+\kappa+2$, contradicting that~$p$ is not contained in any winning $k$-committee of $(C, V)$. As a result, we know that the MAV score of~$w$ can be at most $\abs{B}+d+\kappa+1$. Then, as~$r$ is not dominated by~$B'$, by the construction of the votes, we know that~$w$ contains at least one candidate from~$C(r)$ approved by~$v(r)$. However, in this case the Hamming distance between~$w$ and the vote~$v$ is at least $(\abs{B}+2\kappa+d)-(\kappa-1)+\abs{w\cap C(r)}\geq \abs{B}+\kappa+d+2$, contradicting with the MAV score of~$w$ in $(C, V)$. Hence, we know that $B'$ dominates $R$. Then, as $\abs{B'}\leq \abs{w}=k=\kappa$, we know that the {\prob{RBDS}} instance is a {\yesins.}
\end{proof}

Theorem~\ref{thm-abccv-very-speical-case-np-hard} implies that all the four standard destructive control problems for ABCCV, PAV, and MAV are {\paranph} with respect to the solution size even when there is only one distinguished candidate.
We point out that with respect to~$k$ these problems are fixed-parameter intractable, following from the {\wbhns} (resp.\ {\wahns}) of the winner determination problem for ABCCV and MAV~\cite{DBLP:journals/jair/BetzlerSU13,DBLP:conf/atal/MisraNS15} (resp.\ PAV~\cite{AAMAS2015AzizGMMNW}) and trivial reductions from the winner determination problem to the destructive control problems. In the winner determination problem, we are given an election and a number~$s$, and determine if the election contains a $k$-committee of score smaller than~$s$ (MAV) or larger than~$s$ (ABCCV, PAV). Given an instance of winner determination, we add to the election a set of~$k$ distinguished candidates and add polynomially many new votes so that the $k$-committee containing the

In the following, we derive the complexity of destructive control problems for other rules. Hereinafter, let $m=\abs{C}$ and $n=\abs{V}$.

\subsection{Adding Voters}
{
As {\probb{CCAV}{AV}} is {\nph} and {\wah} with respect to the number of added votes when $k=1$~\cite{DBLP:journals/ai/HemaspaandraHR07,DBLP:journals/tcs/LiuFZL09}, the same holds for {\probb{DCAV}{AV}} when $\abs{J}=m-1$ and $k=1$.
We give a general reduction to show {\wbhns} of \prob{DCAV} for rules satisfying a certain property, for the case $\abs{J}=1$.

\begin{theorem}
\label{thm-dcav-many-rules-wbh}
For every {\abmv} rule~$\varphi$ that is $\alpha$-efficient and neutral,
{\probb{DCAV}{$\varphi$}} is {\wbh} with respect to~$\ell$, the number of added votes. This holds even when there is only one distinguished candidate, $m-k=\bigo{1}$, and there are no registered votes.
\end{theorem}

\begin{proof}
We prove the theorem via a reduction from {\prob{RBDS}}. Let~$\varphi$ be a neutral rule which satisfies $\alpha$-efficiency.
Let $(G, \kappa)$ be an instance of {\prob{RBDS}} where $G=(R\muplus B, \eset)$ is a bipartite graph where both~$R$ and~$B$ are independent sets. %Without loss of generality, assume that $\abs{R}\geq \kappa$ (otherwise, we can solve {\prob{RBDS}} in polynomial-time).
Let $C=R\cup \{p\}$ and $J=\{p\}$. In addition, let $V=\emptyset$, $\ell=\kappa$, and $k=\abs{R}$. For each blue vertex $b\in B$, we create one unregistered vote~$v(b)$ which approves all candidates corresponding to the  neighbors of~$b$ in $G$, i.e., $v(b)=\neighbor{G}{b}$. In the following discussion, for a given $B'\subseteq B$, we use~$U(B')$ to denote the votes created for~$B'$, i.e., $U(B')=\{v(b) \setmid b\in B'\}$. Let $((C, V\muplus U(B)), k, J, \ell)$ be the instance of {\probb{DCAV}{$\varphi$}}, which can be constructed in polynomial time. If there exists $B'\subseteq B$ of at of~$\ell$ vertices so that~$B'$ dominates~$R$, then after adding the unregistered votes corresponding to~$B'$, every red-candidate is approved by at least one vote but~$p$ is not approved by any vote. As $k=\abs{R}$ and~$\varphi$ is $\alpha$-efficient,~$p$ cannot be in any winning $k$-committee of~$\varphi$ at~$(C, U(B'))$. Now we prove for the other direction. Assume that~$B$ does not contain any subset of~$\kappa$ vertices that dominate~$R$. Let~$B'$ be any subset of~$B$ of cardinality at most~$\kappa=\ell$, and let $r\in R$ be a red vertex not dominated by any vertex in~$B'$. By the construction of the unregistered votes,~$r$ is not approved by any vote in~$U(B')$. If~$p$ is not in a winning $k$-committee~$\w$ of~$\varphi$ at~$(C, U(B'))$, it holds that $\w=R$. Then, as~$\varphi$ is neutral, $\w\setminus \{r\}\cup \{p\}$ constitutes another winning $k$-committee of~$\varphi$ at $(C, U(B'))$, implying that the instance of {\probb{DCAV}{$\varphi$}} is a {\noins}.

The above reduction applies to the case where $m-k=1$. By adding~$\bigo{1}$ dummy candidates who are not approved by any vote, it applies then to the case where $m-k=O(1)$.
\end{proof}

{

Recall that AV, SAV, NSAV, and PAV are all $\alpha$-efficient and neutral, but none of the other concrete rules considered in the paper is neutral. We have {\wahns} results  for many sequential Thiele's rules, stated in the following two theorems.

%In the following, we show hardness of {\prob{DCAV}} for rules not covered by Theorem~\ref{thm-dcav-many-rules-wbh}. %, starting with sequential $\omega$-Thiele's rules.
}
\begin{theorem}
\label{thm-dcav-seqpav-nph}
For every $\omega$-Thiele's rule~$\varphi$ such that $\triangle_{\omega}(i)>0$ for all nonnegative integers~$i$, {\probb{DCAV}{\memph{seq}$\varphi$}} is {\wah} with respect to~$\ell$, the number of added votes. This holds even if every candidate is approved by at most two votes, there is only one distinguished candidate, and there are no registered votes.
\end{theorem}

\begin{proof}
We prove the theorem via a reduction from {\prob{MPVC}}.
Let $(G, \kappa, t)$ be an instance of {\prob{MPVC}}.
%Without loss of generality, we assume that $G$ contains at least~$t$ edges.
We create an instance of {\probb{DCAV}{seq$\varphi$}} as follows, where~$\varphi$ is an $\omega$-Thiele's rule such that $\triangle_{\omega}(i)>0$ for all nonnegative integers~$i$.
For every edge $\edge{\vere}{\vere'}\in \eset$, we create one candidate $c(\edge{u}{u'})$. In the following discussion, for a given $\eset'\subseteq\eset$, we use~$C(\eset')$ to denote the set of candidates created for edges in~$\eset'$, i.e., $C(\eset')=\{c(\edge{\vere}{\vere'}) \setmid \edge{\vere}{\vere'}\in \eset'\}$. In addition, we create one distinguished candidate~$p$. Let $J=\{p\}$ and let $C=J \cup C(\eset)$. There are no registered votes, i.e., $V=\emptyset$. The unregistered votes are created according to the vertices of~$G$. In particular, for every $\vere\in \vset$, we create one vote~$v(\vere)$ which approves all candidates corresponding to edges covered by~$\vere$, i.e., $v(\vere)=\{c(\edge{\vere}{\vere'}) \setmid \edge{\vere}{\vere'}\in \eset\}$. In what follows, for a given $\vset'\subseteq \vset$, we use~$U(\vset')$ to denote the set of the unregistered votes corresponding to~$\vset'$, i.e., $U(\vset')=\{v(\vere) \setmid \vere\in \vset'\}$. Finally, we set $k=t$ and $\ell=\kappa$. In the tie-breaking order, the distinguished candidate~$p$ is in the first position. The instance of {\probb{DCAV}{seq$\varphi$}} is $((C, V\muplus U(\vset)), k, J, \ell)$, which can be constructed in polynomial time. According to the tie-breaking order, we know that~$p$ is contained in the winning $k$-committee of seq$\varphi$ at $(C, V)$.
It remains to show the correctness of the reduction.

$(\Rightarrow)$ Assume that there exists $S\subseteq \vset$ of at most $\kappa$ vertices which cover at least~$t$ edges of~$G$. Let $E=(C, U(S))$. Let~$\eset'$ be the set of edges in~$\eset$ covered by~$S$. Then, in the election~$E$, every candidate corresponding to an edge in~$\eset'$ is approved by at least one vote and~$p$ is not approved by any vote. Given $k=t$, $\abs{\eset'}\geq t$, and $\triangle_{\omega}(i)>0$ for all $i\geq 0$, we know that~$p$ cannot be contained in the winning $k$-committee of seq$\varphi$ at~$E$.

$(\Leftarrow)$ Assume that there exists $S\subseteq \vset$ so that $\abs{S}\leq \ell=\kappa$ and the distinguished candidate~$p$ is not contained in the winning $k$-committee of seq$\varphi$ at $E=(C, U(S))$. As~$p$ is the first candidate in the tie-breaking order, we know that there are at least~$k=t$ candidates each of which is approved by at least one vote in~$U(S)$. This directly implies that~$S$ covers at least~$t$ edges in~$G$.
\end{proof}

{Note that seqPAV is a sequential $\omega$-Thiele's rule where $\triangle_{\omega}(i)>0$, but seqABCCV and seqAV are not such rules. We establish another general result which covers these two rules.
}

\begin{theorem}
\label{thm-dcav-seqabccv-wa-hard}
For each $\omega$-Thiele's rule~$\varphi$ such that $\omega(1)>0$ and $\omega(2)\leq 2\omega(1)$,
{\probb{DCAV}{{\memph{seq}$\varphi$}}} is {\wah} with respect to~$k$, the size of a winning committee. Moreover, this holds even if there is only one distinguished candidate, and every vote approves at most two candidates.
\end{theorem}

\begin{proof}
Let~$\varphi$ be an $\omega$-Thiele's rule. Observe that by the definition of $\omega$-Thiele's rules, $\omega(2)<2\omega(1)$ implies already $\omega(1)>0$. In the following, we first prove the {\wahns}  for $\omega$-Thiele's rules such that $\omega(2)<2\omega(1)$, and then we prove for rules such that $\omega(2)=2\omega(1)>0$.

\begin{itemize}
\item $\omega(2)<2\omega(1)$.
\end{itemize}

We establish a reduction from {\prob{Independent Set}} restricted to regular graphs to {\probb{DCAV}{seq$\varphi$}} as follows. Let $(G, \kappa)$ be an instance of {\prob{Independent Set}} where~$G=(\vset, \eset)$ is a regular graph. Let~$d$ be the degree of vertices in~$G$. We create an instance {\probb{DCAV}{seq$\varphi$}} as follows. For every $\vere\in \vset$, we create one candidate denoted by the same symbol for simplicity. In addition, we create one candidate~$p$. Let $J=\{p\}$ and let $C=\vset\cup J$. The tie-breaking order is $\overrightarrow{\vset}\rhd p$. We create a multiset~$V$ of~$d$ registered votes each of which approves only~$p$. %Let~$V$ be the set of the~$d$ registered votes.
Unregistered votes are constructed according to edges in~$G$. Precisely, for each edge $\edge{\vere}{\vere'}\in \eset$, we create one vote $v(\edge{\vere}{\vere'})=\{\vere, \vere'\}$ which approves the two vertex-candidates corresponding to its two endpoints. In the following discussion, for a given  $\eset'\subseteq \eset$, we use~$U(\eset')$ to denote the set of unregistered votes corresponding to~$\eset'$, i.e., $U(\eset')=\{v(\edge{\vere}{\vere'}) \setmid \edge{\vere}{\vere'}\in \eset'\}$.  Note that as every vertex is of degree~$d$, every candidate in $C\setminus \{p\}$ is approved by exactly~$d$ votes in~$U(\eset)$. Finally, we set $k=\kappa$ and $\ell=\kappa\cdot d$. The instance of {\probb{DCAV}{seq$\varphi$}} is $((C, V\muplus U(\eset)), k, J, \ell)$, which can be constructed in polynomial-time. Due to~$\omega(1)>0$, we know that~$p$ is contained in the winning $k$-committee of seq$\varphi$ at~$(C, V)$. In the following, we prove that the instance of {\prob{Independent Set}} is a {\yesins} if and only if the constructed instance of {\probb{DCAV}{seq$\varphi$}} is a {\yesins}.

$(\Rightarrow)$ Suppose that~$G$ admits an independent set $\vset'\subseteq \vset$ of size~$\kappa$. Let $\partial(\vset')=\{\edge{\vere}{\vere'}\in \eset \setmid \{\vere, \vere'\}\cap \vset'\neq\emptyset\}$ be the set of edges covered by~$\vset'$. As~$G$ is $d$-regular and~$\vset'$ is an independent set, it holds that $\abs{\partial(\vset')}=\kappa\cdot d=\ell$.
Let $E=(C, V\cup U(\partial(\vset')))$, and let~$\w$ denote the winning $k$-committee of seq$\varphi$ at~$E$. It is easy to verify that~$p$ is not contained~$\w$. In fact, by the construction of the unregistered votes and the fact that~$\vset'$ is an independent set, every candidate $\vere\in \vset'$ is approved by exactly~$d$ votes in~$U(\partial(\vset'))$ and, moreover none of these~$d$ votes approves anyone else in~$\vset'$. Though~$p$ is also approved by~$d$ votes in~$E$, it is ranked in the last position in the tie-breaking order. Therefore,~$p$ can be contained in~$\w$ only if all of~$\vset'$ are contained in~$\w$. However, as $\abs{\vset'}=\kappa=k$, it must be that~$p\not\in \w$.

$(\Leftarrow)$ Suppose that there exists $\eset'\subseteq \eset$  so that $\abs{\eset'}\leq \ell$ and~$p$ is not contained in the winning $k$-committee, denoted~$\w$, of seq$\varphi$ at $E=(C, V\cup U(\eset'))$. We prove below that~$\w$ is an independent set in~$G$. To this end, let us consider the order~$\w^1_E$,~$\w^2_E$, $\dots$,~$\w^k_E$ of~$\w$. To show that~$\w$ is an independent set in~$G$, it suffices to show that $w^i_E$ is not adjacent to anyone from $\w^{\leq i-1}_E$ in the graph~$G$ for every $i\in [k]\setminus \{1\}$. Let~$i$ be an integer in $[k]\setminus \{1\}$. As there are~$d$ votes in~$E$ approving only~$p$, and there are no other votes approving~$p$, the~$\varphi$ margin contribution of~$p$ to~$\w^{\leq i-1}_E$ is $d\cdot \omega(1)$. Then, as $p\not\in \w$ and $\omega(1)>0$, the~$\varphi$ margin contribution of~$\w^i_E$ to~$\w^{\leq i-1}_E$ is at least~$d \cdot \omega(1)$. As $\w^i_E$ is approved by exactly~$d$ votes in~$U(\eset)$ (corresponding to the~$d$ edges covered by~$\w^i_E$) and is not approved by any vote from~$V$, we know that all the~$d$ votes approving~$\w^i_E$ are contained in~$U(\eset')$ and, more importantly, by $\omega(2)<2\omega(1)$ it holds that none of these~$d$ votes approves any candidates in~$\w^{\leq i-1}_E$ (otherwise the~$\varphi$ margin contribution of~$\w^i_E$ to~$\w^{\leq i-1}_E$ is strictly smaller than $d\cdot \omega(1)$). This directly implies that~$\w^i_E$ is not adjacent to any of~$\w^{\leq i-1}_E$ in the graph~$G$. As this holds for all $i\in [k]\setminus \{1\}$, we conclude that~$\w$ forms an independent set in~$G$.

\begin{itemize}
\item $\omega(2)=2\omega(1)>0$.
\end{itemize}

%To prove the {\wahns} of {\probb{CDAV}{seq$\varphi$}} for an $\omega$-Thiele's rule such that $\omega(2)=2\omega(1>0)$,
We establish a reduction from {\prob{Clique}} restricted to regular graphs to {\probb{CDAV}{seq$\varphi$}}. The reduction is similar to the one just studied above. In particular, given an instance $(G, \kappa)$ of {\prob{Clique}} where~$G=(\vset,\eset)$ is a $d$-regular graph,
we create the same instance $((C, V\muplus U(\eset)), k, J, \ell)$ as constructed above with only the difference that we set $\ell=\kappa \cdot d - \frac{\kappa\cdot (\kappa-1)}{2}$.  We assume that $d\geq \kappa-1$ since otherwise the {\prob{Clique}} instance is a trivial {\noins}.  It remains to prove for the correctness of the reduction.

$(\Rightarrow)$ Suppose that~$G$ admits a clique $\vset'\subseteq \vset$ of size~$\kappa$. Let $\partial(\vset')$ be the set of edges covered by~$\vset'$ as defined in the above reduction.
%, i.e., $\partial(\vset')=\{\edge{\vere}{\vere'}\in \eset \setmid \{\vere, \vere'\}\cap \vset'\neq\emptyset\}$.
As~$G$ is $d$-regular and~$\vset'$ is a clique, it holds that $\abs{\partial(\vset')}=\kappa \cdot d - \frac{\kappa\cdot (\kappa-1)}{2}=\ell$.
Let $E=(C, V\cup U(\partial(\vset')))$, and let~$\w$ denote the winning $k$-committee of seq$\varphi$ at~$E$.
We show below that $p\not\in \w$. To this end, let us consider the sequence $\w^1_E$, $\w^2_E$, $\dots$, $\w^k_{E}$ of~$\w$.
By the construction of the unregistered votes, every candidate $\vere\in \vset'$ is approved by exactly~$d$ votes in~$U(\partial(\vset'))$.
Though~$p$ is also approved by~$d$ votes in~$E$,~$p$ is ranked last in the tie-breaking order.
Therefore, it holds that $\w^1_E\in \vset'$. Assuming $\w^{\leq i}\subseteq \vset'$ for $i\leq k-1$, we prove that $\w^{i+1}_E\in \vset'$.
To see this, recall $\triangle_{\omega}(1)=\omega(2)-\omega(1)=\omega(1)>0$. For every $\vere\in \vset'$, let~$x_{\vere}$ be the number of neighbors of~$\vere$ contained in~$\w^{\leq i}_E$, i.e., $x_{\vere}=\abs{\neighbor{G}{\vere}\cap \w^{\leq i}_E}$.
Then, the~$\varphi$ margin contribution of every $\vere\in \vset'\setminus \w^{\leq i}_E$ to~$\w^{\leq i}_E$ is
$\triangle_{\omega}(1) \cdot x_{\vere} + \omega(1)\cdot (d-x_{\vere})=d\cdot \omega(1)$, which equals that of~$p$ to~$\w^{\leq i}_E$. As~$p$ is after~$\vere$ in the tie-breaking order, we know that~$p$ cannot be~$\w^{i+1}_E$. We can deduce then that $p\not\in \w$.

$(\Leftarrow)$ Assume that there exists~$\eset'\subseteq \eset$ such that $\abs{\eset'}\leq \ell$ and~$p$ is not contained in the winning $k$-committee, denoted~$\w$, of seq$\varphi$ at the election $E=(C, V\cup U(\eset'))$. Hence, it holds~$\w\subseteq \vset$. We claim first that every candidate in~$\w$ is approved by~$d$ votes from~$U(\eset')$. To prove the claim, for the sake of contradiction, assume that~$\w$ contains a candidate which is approved by less than~$d$ votes from~$U(\eset')$. Let~$i$ be the smallest integer~$i$ so that~$\w^i_E$ is approved by at most $d-1$ votes from~$U(\eset')$. Then, as $\triangle_{\omega}(1)=\omega(1)$, the~$\varphi$ margin contribution of~$\w^i_E$ to $\w^{\leq i-1}_E$ can be at most $(d-1) \cdot \omega(1)$. However, the~$\varphi$ margin contribution of~$p$ to~$\w^{\leq i}_E$ is $d\cdot \omega(1)$, and by $\omega(1)>0$ this contradicts with~$\w^i_E\in \vset$. This completes the proof of the claim.
The claim implies that for every $\vere\in\w$, all the~$d$ edges covered by~$\vere$ are contained in~$\eset'$. In other words, $\partial(\w)\subseteq \eset'$. As a consequence, $\abs{\partial(\w)}\leq \ell=\kappa \cdot d - \frac{\kappa\cdot (\kappa-1)}{2}$. As~$G$ is a $d$-regular graph, this can be the case only when $\abs{\partial(\w)}= \ell$ and~$\w$ forms a clique. In fact, for the sake of contradiction, assume that~$\w$ is not a clique. Let~$t$ denote the number of missing edges between vertices in~$\w$. So, it holds that $t\geq 1$. As~$G$ is $d$-regular, we have that $\abs{\partial({\vset'})}=\kappa\cdot d-(\frac{\kappa \cdot (\kappa-1)}{2}-t)>\ell$, a contradiction.
\end{proof}

\onlyfull
{
As {\prob{MPVC}} remains {\nph} even when restricted to graphs of maximum degree three and when~$\abs{\eset}-t$ is a constant, the above proof gives the following result.

\begin{corollary}
For every $\omega$-Thiele's rule~$\varphi$ such that $\omega(1)>0$, {\probb{DCAV}{\memph{seq}$\varphi$}} is {\paranph} with respect to the parameter~$m-k$, where~$m$ is the number of candidates.
Moreover, this holds even when every voter approves at most three candidates, every candidates is approved by at most two voters, and there is only one distinguished candidate.
\end{corollary}
}

As a natural parameter, it is important to see if for $\varphi\in \{\text{AV}, \text{SAV}, \text{NSAV}\}$, {\probb{DCAV}{$\varphi$}} remains fixed-parameter intractable with respect to~$k$. We show below that the problem is in fact {\wah} even when combined parameterized by~$k$,~$\ell$, and the number of registered votes, for all AV-uniform rules~$\varphi$. Recall that AV, SAV, and NSAV are AV-uniform.

\begin{theorem}
\label{thm-dcav-many-rules-wbh-k-ell}
For each AV-uniform {\abmv} rule~$\varphi$,
{\probb{DCAV}{$\varphi$}} is {\wah} with respect to~$k+\ell+n_{\memph{rg}}$, where~$n_{\memph{rg}}$ denotes the number of registered votes. This holds even when there is only one distinguished candidate and every vote approves two candidates.
\end{theorem}

\begin{proof}
Let~$\varphi$ be an AV-uniform {\abmv} rule.
    We prove Theorem~\ref{thm-dcav-many-rules-wbh-k-ell} via a reduction from {\prob{Clique}}. Let $(G, \kappa)$ be an instance of {\prob{Clique}}, where $G=(\vset, \eset)$. We construct an instance of {\probb{DCAV}{$\varphi$}} as follows. The candidate set is $C=\vset\cup \{p, p'\}$, where~$p$ is the only  distinguished candidate. We create a multiset~$V$ of $\kappa-2$ registered votes each of which approves exactly~$p$ and~$p'$. For each edge $\{\vere, \vere'\}$ of~$G$, we create one unregistered vote $v(\vere, \vere')$ approving exactly~$\vere$ and~$\vere'$. Let~$U$ be the set of the unregistered votes. Finally, let $k=\kappa$ and $\ell=\frac{\kappa \cdot (\kappa-1)}{2}$. The instance of {\probb{DCAV}{$\varphi$}} is $((C, V\cup U), k, \{p\}, \ell)$. Obviously, every vote in $V\cup U$ approves exactly two candidates. In the following, we prove that the {\prob{Clique}} instance is a {\yesins} if and only if the {\probb{DCAV}{$\varphi$}} instance is a {\yesins}.

    $(\Rightarrow)$ Assume that~$G$ contains a clique~$K$ of size~$\kappa$. Let~$U(K)$ be the set of the $\frac{\kappa \cdot (\kappa-1)}{2}$ votes corresponding to the edges among vertices of~$K$ in~$G$. Let~$E=(C, V\cup U(K))$. By the above construction of the votes, in the election~$E$, every candidate from~$K$ receives $\kappa-1$ approvals, the distinguished candidate~$p$ receives $\kappa-2$ approvals, and every other candidate receives at most one approval. As $\abs{K}=k$, we know that~$p$ cannot be in any winning $k$-committees of~$\varphi$ at~$E$, and hence the {\probb{DCAV}{$\varphi$}} instance is a {\yesins}.

    $(\Leftarrow)$ Assume that there is a subset~$U'$ of at most~$\ell$ votes so that $p$ is not contained in any winning $k$-committees of the election $E=(C, V\cup U')$. As none of the unregistered votes approves~$p$ and~$p'$, we know that~$p$ receives~$\kappa-2$ approvals in~$E$. This implies that there are at least~$k$ candidates from~$\vset$ each receiving at least $\kappa-1$ approvals in~$E$. Let~$B$ denote the set of all candidates receiving at least~$\kappa-1$ approvals in~$E$. As every vote approves exactly two candidates, it holds that $\abs{B}=\kappa$, $\abs{U'}= \frac{\kappa \cdot (\kappa-1)}{2}$ and, moreover,~$K$ is a clique of~$G$. So, the {\prob{Clique}} instance is a {\yesins}.
\end{proof}

\subsection{Deleting Voters}
Now we establish a reduction for {\prob{DCDV}} that applies to a class of {\abmv} rules. It should be mentioned that Magiera~\citeas{Magierphd2020} obtained {\nphns} of {\probb{DCDV}{AV}} and {\probb{DCDV}{SAV}} via two separate reductions from {\prob{Exact Cover by Three Sets}}.
However, the reductions do not imply {\wbhns} with respect to any parameters studied in the paper.
The {\wbhns} result of {\probb{CCDV}{AV}} by~\citeas{DBLP:journals/tcs/LiuFZL09} indicates that {\probb{DCDV}{AV}} is {\wbh} with respect to the number of deleted votes when $\abs{J}=m-1$ and $k=1$, but their reduction is only tailored for the AV rule. We derive a reduction to show the {\wbhns} of {\probb{DCDV}{$\varphi$}} for all AV-uniform rules for the special case where there is only one distinguished candidate.

\begin{theorem}
\label{thm-dcdv-save-wbh-remaining-votes}
For every AV-uniform {\abmv} rule~$\varphi$,
{\probb{DCDV}{$\varphi$}} is {\wbh} with respect to~$n-\ell$, the number of votes not deleted. This holds even when there is only one distinguished candidate. % and $k=m-1$ where~$m$ is the number of candidates.
\end{theorem}

\begin{proof}
Let~$\varphi$ be an AV-uniform {\abmv} rule.
We prove the theorem via a reduction from {\prob{RBDS}}. Let $(G, \kappa)$ be an instance of {\prob{RBDS}} where $G=(R\muplus B, \eset)$ is a bipartite graph with~$R$ and~$B$ being two independent sets.
Let~$d=\max_{b\in B}\degree{G}{b}$ be the maximum degree of vertices in~$B$ in the graph~$G$, and let $d'=\max_{r\in R}\degree{G}{r}$ be the maximum degree of vertices in~$R$.
We assume $\abs{R}\leq \kappa \cdot (d+1)$ and $\abs{B}\geq \kappa+d'$. This assumption does not change the {\wbhns} of {\prob{RBDS}}.\footnote{If $\abs{R}>\kappa\cdot (d+1)$, the {\prob{RBDS}} instance is a trivial {\noins}. If $\abs{B}< \kappa+d'$, then any red vertex with degree~$d'$ is dominated by any subset of~$\kappa$ blue vertices. Therefore, if this condition is not satisfied, we can preprocess the given instance by iteratively removing red vertices with the maximum degree until the condition is satisfied. By the above discussion, the original instance and the instance after the preprocessing are equivalent.}
%Let $t \equiv \abs{R} \pmod{d}$.
Assume that $\abs{R}=\kappa'\cdot (d+1)+t$ where~$\kappa'$ and~$t$ are two nonnegative integers so that $\kappa'\leq \kappa$ and $t<d+1$. As $\abs{R}\leq \kappa \cdot (d+1)$, such~$\kappa'$ and~$t$ exist and they are unique.
We create an instance of {\probb{DCDV}{$\varphi$}} as follows.

The candidates are constructed as follows.
\begin{itemize}
\item For each $r\in R$, we create one candidate~$c(r)$. Let $C(R)=\{c(r) \setmid r\in R\}$ be the set of candidates corresponding to~$R$.
\item We create a set~$X$ of~$d+1-t$ candidates.
%\item We create a set~$Y$ of~$d$ candidates.
\item For every $b\in B$, we create a set~$C(b)$ of $d-\degree{G}{b}$ candidates. Note that if~$b$ has degree~$d$ in~$G$, it holds that $C(b)=\emptyset$. Let $C(B)=\bigcup_{b\in B}C(b)$ be the set of all candidates created for vertices in~$B$.
\item Finally, we create a candidate~$p$ which is the only distinguished candidate.
\end{itemize}
Let $J=\{p\}$ and let $C=C(R)\cup X\cup J\cup C(B)$.
We divide candidates in $C(R)\cup X$ into~$\kappa'+1$ pairwise disjoint subsets $X_1$, $X_2$, $\dots$, $X_{\kappa'+1}$ so that each subset consists of exactly~$d+1$ candidates.

We create the following votes.
\begin{itemize}
%    \item First, for each $r\in R$, we create a multiset %~$V(r)$
%    of~$\kappa$ votes each of which approves exactly the $d+1$ candidates in $\{c(r)\}\cup Y$.
%    \item In addition, for each $x\in X$, we crate one vote %~$v(x)$
%    which approves exactly the~$d+1$ candidates in~$\{x\}\cup Y$.
    \item For each $i\in [\kappa'+1]$, we create one multiset of $\kappa$ votes each of which approves exactly the~$d+1$ candidates in~$X_i$.

    \item For each blue vertex $b\in B$, we create one vote~$v(b)=\{p\}\cup C(\neighbor{G}{b})\cup C(b)$ approving~$p$, all candidates corresponding to neighbors of~$b$ in~$G$, and candidates created for~$b$. Observe that every~$v(b)$ approves exactly $d+1$ candidates. In the following discussion, for a given $B'\subseteq B$, we use~$V(B')$ to denote the set of votes corresponding to~$B'$, i.e., $V(B')=\{v(b) \setmid b\in B'\}$.
\end{itemize}
Let~$V$ be the multiset consisting of the above created $\kappa \cdot (\kappa'+1)+\abs{B}$ votes.
Finally, we set $\ell=\abs{B}-\kappa$ and $k=\abs{R}$.
The instance of {\probb{DCDV}{$\varphi$}} is $((C, V), k, J, \ell)$.
In the election~$(C, V)$,
%candidates in~$Y$ have AV score $(\abs{B}+1)\cdot (d+1)+\kappa\cdot \abs{R}$, which is large enough to ensure them to be included in all winning $k$-committees, even after deleting at most~$\ell$ votes. In addition,
every candidate~$c(r)$ where $r\in R$ has AV score $\kappa+\degree{G}{r}$, the distinguished candidate~$p$ has AV score~$\abs{B}$, and every other candidate has AV score at most~$\kappa$. Notice that as $\abs{B}\geq \kappa+d'$,~$p$ has at least the same AV score than that of every candidate~$c(r)\in C(R)$, and hence~$p$ is contained in at least one winning $k$-committee of~$\varphi$ at $(C, V)$.
 In the following, we prove the correctness of the reduction. %that the {\prob{RBDS}} instance is a {\yesins} if and only if the above constructed instance of {\probb{DCDV}{$\varphi$}} is a {\yesins}.

$(\Rightarrow)$ Assume that there exists~$B'\subseteq B$ such that $\abs{B'}= \kappa$ % (if~$B$ contains a subset of at most~$\kappa$ vertices dominating~$R$, then it contains also a subset of exactly~$\kappa$ vertices dominating~$R$)
and~$B'$ dominates~$R$. Let $\overline{B}=B\setminus B'$, and let $E=(C, V\setminus V(\overline{B}))$. In the election~$E$, every candidate in~$X$ has the same AV score~$\kappa$ as in $(C, V)$, and~$p$ has AV score~$\kappa$. As~$B'$ dominates~$R$, by the construction of votes, for every $r\in R$, there exists at least one vote in $V(B')=V(B)\setminus V(\overline{B})$ which approves the candidate $c(r)$. Therefore, every $c(r)\in C(R)$ where $r\in R$ has AV score at least $\kappa+1$ in~$E$. Then, as $\abs{C(R)}=k$, we know that~$p$ is not contained in any winning $k$-committee of~$\varphi$ at~$E$.

$(\Leftarrow)$ Assume that there exists $V'\subseteq V$ of at most~$\ell$ votes so that~$p$ is not contained in any winning $k$-committee of~$\varphi$ at the election $E=(C, V\setminus V')$. We may assume that~$V'$ contains only votes approving the distinguished candidate~$p$, i.e., $V'\subseteq V(B)$. The reason is that if~$V'$ contained a vote disapproving~$p$, then removing this vote from~$V'$ into $V\setminus V'$ only increases the AV scores of other candidates without changing the AV score of~$p$; and hence~$p$ still remains excluded from any winning $k$-committee of~$\varphi$ at the resulting election.
Without loss of generality, let $V'=V(B')$ for some $B'\subseteq B$, and let $B''=B\setminus B'$. Besides, let $\abs{B''}=\kappa+z$ for some~$z\geq 0$.  Then, the AV score of~$p$ in~$E$ is $\kappa+z$. Note that in the election~$E$ the AV score of every candidate in~$X$ remains the same as in~$(C,V)$, and every candidate in~$C(B)$ has AV score at most one. As~$p$ is not contained in any winning $k$-committee of~$\varphi$ at~$E$, we know that every candidate $c(r)\in C(R)$ where $r\in R$ in~$E$ has AV score at least $\kappa+z+1$. By the construction of the votes, this indicates that every $r\in R$ is dominated by at least~$z+1$ vertices in~$B''$. Then, after removing any arbitrary~$z$ vertices from~$B''$, we obtain a $\kappa$-subset of~$B$ that dominates~$R$.
\end{proof}

For the two natural parameters~$k$ and~$\ell$, we have the results summarized in the following two theorems.

\begin{theorem}
\label{thm-dcdv-av-wah-k-ell}
For every {\abmv} AV-uniform rule~$\varphi$,  {\probb{DCDV}{$\varphi$}} is {\wah} with respect to~$k+\ell$. This holds even if there is only one distinguished candidate.
\end{theorem}

\begin{proof}
Let~$\varphi$ be an AV-uniform {\abmv} rule.
    We prove Theorem~\ref{thm-dcdv-av-wah-k-ell} via a reduction from {\prob{Clique}} restricted to regular graphs. Let $(G, \kappa)$ be an instance of {\prob{Clique}}, where $G=(\vset, \eset)$ is a graph where the degree of every vertex is~$d$. Without loss of generality, we assume that $\abs{\eset} \geq d \geq \kappa\geq 3$. We create the following candidates. First, for each vertex of~$G$, we create one vertex denoted by the same symbol for simplicity. Then, for each edge $\edee\in \eset$, we create one candidate~$c(\edee)$. Besides, we create a distinguished candidate~$p$. Let~$C=\vset\cup \{c(\edee) \setmid \edee\in \eset\}\cup \{p\}$, and let $J=\{p\}$. We create the following votes.
    \begin{itemize}
        \item For each edge $\edee=\{\vere, \vere'\}$, we create one vote~$v(\edee)=(\vset\cup \{p, c(\edee)\})\setminus \{\vere, \vere'\}$ approving all candidates from~$\vset\cup \{p, c(\edee)\}$ except~$\vere$ and~$\vere'$.
        \item In addition, we create $d+2-\kappa$ votes each approving all candidates in~$\vset$.
    \end{itemize}
    Let~$V$ be the multiset of the above created votes. Obviously, every vote approves exactly~$\abs{\vset}$ candidates. Finally, let $k=\kappa$ and let $\ell=\frac{\kappa \cdot (\kappa-1)}{2}$.
    The instance of {\probb{DCDV}{$\varphi$}} is $((C, V), k, J, \ell)$, which clearly can be constructed in polynomial time. It is easy to verify that with respect to~$V$, every candidate from~$\vset$ receives exactly $(\abs{\eset}-d)+(d+2-\kappa)=\abs{\eset}+2-\kappa$ approvals,~$p$ receives~$\abs{\eset}$ approvals, and every other candidate receives one approval. As $\kappa\geq 3$, we know that~$p$ is contained in all winning $k$-committees of $(C, V)$.
    We prove the correctness of the reduction as follows.

    $(\Rightarrow)$ Assume that~$G$ contains a clique~$K$ of~$\kappa$ vertices. Let~$V'$ be the set of the $\frac{\kappa \cdot (\kappa-1)}{2}$ votes corresponding to the edges among vertices in~$K$. Let~$E=(C, V\setminus V')$. By the construction of the votes and the fact that~$K$ is a clique, we know that every candidate from~$K$ receives $(\abs{\eset}+2-\kappa)-(\frac{\kappa \cdot (\kappa-1)}{2}-(\kappa-1))=\abs{\eset}-\frac{\kappa \cdot (\kappa-1)}{2}+1$ approvals, the distinguished candidate~$p$ receives~$\abs{\eset}-\frac{\kappa \cdot (\kappa-1)}{2}$ approvals. As $k=\kappa=\abs{K}$, we know that~$p$ cannot be in any winning $k$-committees of~$\varphi$ at~$E$.

    $(\Leftarrow)$ Assume that there is a submultiset $V'\subseteq V$ of~$\ell'$ votes so that $\ell'\leq \ell$ and~$p$ is not contained in any winning $k$-committee of~$\varphi$ at $E=(C, V\setminus V')$. This implies that there are at least~$k$ candidates from~$C\setminus \{p\}$ each receiving more approvals than that of~$p$ in~$E$. Observe that these candidates must be from $\vset$. Without loss of generality, let~$B$ denote the set of candidates from~$C\setminus \{p\}$ each receiving more approvals than~$p$ in~$E$. As discussed above, $B\subseteq \abs{\vset}$ and $\abs{B}\geq k$. We also observe that if the {\probb{DCDV}{$\varphi$}} instance is a {\yesins}, it admits a feasible solution where all deleted votes approve~$p$, i.e., all votes in the feasible solution are from those corresponding to edges of~$G$. By this observation, we may assume that~$p$ receives $\abs{\eset}-\ell'$ approvals in~$E$. As every candidate~$c\in B$ receives $\abs{\eset}+2-\kappa$ approvals from~$V$, and the number of approvals of~$c$ received from $V\setminus V'$ drops to a value larger than $\abs{\eset}-\ell'$, the total number of approvals of candidates in~$B$ received from~$V'$ can be at most $\abs{B}\cdot ((\abs{\eset}+2-\kappa)-(\abs{\eset}-\ell'+1))=\abs{B}\cdot (\ell'-\kappa+1 )$. By the construction of the votes, every vote from~$V'$ approves at least $\abs{B}-2$ candidates in~$B$. As a consequence, it holds that
    \begin{equation}
        \label{eq-aaa}
        \ell' \cdot (\abs{B}-2) \leq \abs{B} \cdot (\ell'-\kappa+1).
    \end{equation}
    By the fact that $\abs{B}\geq \kappa$ and $\ell'\leq \ell=\frac{\kappa \cdot (\kappa-1)}{2}$, one can check easily by elementary calculation that the only possibility so that Inequality~\ref{eq-aaa} holds is when $\abs{B}=\kappa$ and $\ell'=\ell$. In fact, in this case it holds that $\ell' \cdot (\abs{B}-2) \leq \abs{B} \cdot (\ell'-\kappa+1)$. Moreover, every vote in~$V'$ approves exactly $\kappa-2$ candidates from~$B$. This also implies  that~$B$ is the set of vertices spanned by edges corresponding to~$V'$. Then, given $\ell'=\ell=\frac{\kappa \cdot (\kappa-1)}{2}$, it follows that~$B$ is a clique of~$k$ vertices in~$G$.
\end{proof}

Theorems~\ref{thm-dcdv-save-wbh-remaining-votes}--\ref{thm-dcdv-av-wah-k-ell} imply that {\prob{DCDV}} for AV, SAV, and NSAV are {\wbh} with respect to the number of votes not deleted, and are {\wah} when jointly parameterized by the number of deleted votes and the size of the winning committee.

{
For AV, we further show that with respect to the dual parameter $m-k$ of $k$, i.e., the number of candidates not in a desired winning committee, {\prob{DCDV}} is {\paranph}.
%Besides the above significant extension (Theorem~\ref{thm-dcdv-save-wbh-remaining-votes}) of these results, we also enhance the previous results for AV as follows.

\begin{theorem}
\label{thm-dcdv-av-wbh-remaining-votes}
{\probb{DCDV}{{\memph{AV}}}} is {\paranph} when parameterized by $m-k$.
\end{theorem}

\begin{proof}
We prove the hardness of {\probb{DCDV}{{\memph{AV}}}}  for the case where $m-k=1$.  Let $(G, \kappa)$ be an instance of {\prob{RBDS}} where $G=(R\muplus B, \eset)$ is a bipartite graph.
%We assume that every red vertex has at least one neighbor and at most $\abs{B}-1$ neighbors. In addition, we assume that $\kappa+\ell\leq \abs{R}$. (If this is not the case, we can add $\ell$ new blue vertices and as many red vertices as we want so that the newly added red vertices dominate exactly the newly added $\ell$ blue vertices. Clearly, the original graph has a red blue dominating set of size $\kappa$ if and only if the new graph has a red blue dominating set of size $\kappa+1$).
We create an instance {\probb{DCDV}{AV}} as follows. Let $C=R\cup \{p\}$, $J=\{p\}$, $k=\abs{R}$, and $\ell=\abs{B}-\kappa$. We create the following votes:
\begin{itemize}
    \item a multiset~$\widetilde{V}$ of~$\kappa$ votes each of which approves exactly the candidates in~$R$;
    \item for each blue vertex $b\in B$, one vote~$v(b)$ which approves~$p$ and all candidates corresponding to neighbors of~$b$ in~$G$, i.e., $v(b)=\{p\}\cup \neighbor{G}{b}$.
\end{itemize}
In the following analysis, for a given $B'\subseteq B$, we use~$V(B')$ to denote the multiset of votes corresponding to~$B'$, i.e., $V(B')=\{v(b) \setmid b\in B'\}$. Let $V=\widetilde{V}\cup V(B)$. In total, we have $\abs{B}+\kappa$ votes.
The instance of {\probb{DCDV}{AV}} is $((C, V), k, J, \ell)$.
Now we prove the correctness of the reduction.

$(\Rightarrow)$ Suppose that there exists $B'\subseteq B$ of~$\kappa$ vertices dominating~$R$. Let $E=(C, \widetilde{V}\cup V(B'))$. By the construction of votes, for every $r\in R$, there exists at least one vote in~$V(B')$ which approves~$r$. Together with the~$\kappa$ votes in~$\widetilde{V}$, there are in total at least~$\kappa+1$ votes in the election~$E$ which approve~$r$. The distinguished candidate~$p$ is only approved by votes in~$V(B')$. From $\abs{V(B')}=\abs{B'}= \kappa$ and $k=\abs{R}$, we know~$R$ is the unique winning $k$-committee of AV at~$E$. As $p\not\in R$, the instance of {\probb{DCDV}{AV}} is a {\yesins}.

$(\Leftarrow)$ Suppose that there exists $V'\subseteq V$ such that $\abs{V'}\leq \ell=\abs{B}-\kappa$ and~$p$ is not contained in any winning $k$-committee of AV at $E=(C, V\setminus V')$. We may assume that $\widetilde{V}$ and~$V'$  are disjoint, because it is easy to verify that~$p$ is also not contained in any winning $k$-committee of AV at $(C, V\setminus V'\cup \widetilde{V})$. So, let $V'=V(B')$ for some $B'\subseteq B$. Assume, without loss of generality, that~$\abs{V'}=\ell-t$ for some nonnegative integer~$t$. Then, there are exactly $\kappa+t$ votes approving~$p$ in~$E$, and hence the AV score of every other candidate in~$E$ must be at least $\kappa+t+1$. This indicates that for every $r\in R$ there are at least $t+1$ votes in $V(B)\setminus V'$ approving~$r$, or equivalently, every $r\in R$ has at least $t+1$ neighbors from $B\setminus B'$ in~$G$. Let~$B''$ be the subset obtained from $B\setminus B'$ by removing any arbitrary~$t$ vertices. In light of the above discussion, we have that $\abs{B''}=\kappa$ and, more importantly, for every $r\in R$, there exists at least $(t+1)-t=1$ vertex in~$B''$ which dominates~$r$. This means that the {\prob{RBDS}} instance is  a {\yesins}.

To prove the hardness of {\probb{DCDV}{{\memph{AV}}}}  for the case where $m-k=\bigo{1}$, we adapt the above reduction by adding~$\bigo{1}$ candidates who are not approved by any votes.
\end{proof}
}

For sequential $\omega$-Thiele's rules, we have a {\wbhns} result with respect to the solution size.

\begin{theorem}
\label{thm-dcdv-seq-wbh}
For every $\omega$-Thiele's rule~$\varphi$ such that $\omega(1)>0$, {\probb{DCDV}{\memph{seq}$\varphi$}} is {\wbh} with respect to~$\ell$, the number of deleted votes.
\end{theorem}

\begin{proof}
We prove the theorem via a reduction from {\prob{RBDS}}. Let~$\varphi$ be an $\omega$-Thiele's rule such that $\omega(1)>0$. Let $(G, \kappa)$ be an instance of {\prob{RBDS}} where~$G$ is a bipartite graph with two disjoint independent sets~$R$ and~$B$. As in the proof of Theorem~\ref{thm-ccac-seqThiele-wbh}, we assume that~$G$ does not contain any isolated vertices, and every red vertex $r\in R$ has degree exactly~$d$ for some positive integer~$d$.
We construct an instance of {\probb{DCDV}{seq$\varphi$}} based on the instance of {\prob{RBDS}} as follows. For every $r\in R$, we create one candidate denoted by the same symbol for simplicity. In addition, we create a set~$X$ of~$\abs{R}$ candidates. Let $C=R\cup X$, let $J=R$, and let $k=\abs{R}$.
%, i.e., we aim to select winning committees of~$\abs{R}$ candidates.
The linear order for  tie-breaking is $\overrightarrow{R}\rhd \overrightarrow{X}$. We create the following $\abs{B}+d\cdot \abs{R}$ votes:
\begin{itemize}
    \item for each $b\in B$, we create one vote $v(b)=\neighbor{G}{b}$ approving candidates corresponding to its neighbors in~$G$;
    \item for each $x\in X$, we create a multiset $V(x)$ of~$d$ votes each of which approves only the candidate~$x$.
\end{itemize}
In the following exposition, for a given $B'\subseteq B$, we use~$V(B')$ to denote the set of votes created for vertices in~$B'$, i.e., $V(B')=\{v(b) \setmid b\in B'\}$. In addition, let $V(X)=\bigcup_{x\in X}V(x)$, and let $V=V(B)\cup V(X)$.
Note that as every red vertex has exactly~$d$ neighbors in~$G$, every $r\in R$ is approved by exactly~$d$ votes in~$V$. By the definition of the tie-breaking order~$\rhd$, the first candidate added into the winning $k$-committee of seq$\varphi$ at~$(C, V)$ is from~$R$.
Finally, let $\ell=\kappa$, i.e., we are allowed to delete at most~$\kappa$ votes from~$V$.
The instance of {\probb{DCDV}{seq$\varphi$}} is $((C, V), k, J, \ell)$ which clearly can be created in polynomial time.
In the following, we show that the {\prob{RBDS}} instance is a {\yesins} if and only if the above constructed instance of {\probb{DCDV}{seq$\varphi$}} is a {\yesins}.

$(\Rightarrow)$ Assume that there exists $B'\subseteq B$ of~$\kappa$ vertices dominating~$R$. Let $E=(C, V\setminus V(B'))$. Let~$\w$ be the winning $k$-committee of seq$\varphi$ at~$E$. We show below that $\w=X$, implying that the instance of {\probb{DCDV}{seq$\varphi$}} is a {\yesins}. In fact, as~$B'$ dominates~$R$, for every $r\in R$ there exists at least one vote in~$V(B')$ which approves~$r$. Therefore, there are at most $d-1$ votes in $V\setminus V(B')$ approving~$r$. As every candidate $x\in X$ is approved by exactly~$d$ votes which approve only~$x$, $\abs{X}=\abs{R}=k$, and $\omega(1)>0$, we know that the winning $k$-committee of seq$\varphi$ at~$E$ is~$X$. This completes the proof for this direction.

$(\Leftarrow)$ Suppose that there exists $V'\subseteq V$ of at most~$\kappa$ votes so that none of~$J$ is contained in the winning $k$-committee of seq$\varphi$ at $E=(C, V\setminus V')$. Given $k=\abs{C}-\abs{R}$, we know that the winning $k$-committee of seq$\varphi$ at~$E$ is~$X$. As a consequence, we may assume that~$V'$ does not contain any vote in~$V(X)$. As a matter of fact, if~$V'$ contained a vote from~$V(x)$ for some $x\in X$, then moving this vote from~$V'$ into $V\setminus V'$ does not refrain~$X$ from being the winning $k$-committee of seq$\varphi$ at~$(C, V\setminus V')$. Then, without loss of generality, let $V'=V(B')$ for some $B'\subseteq B$. We show below that~$B'$ dominates~$R$. For the sake of contradiction, assume that there exists $r\in R$ not dominated by any vertex in~$B'$. This means that none of the~$d$ votes approving~$r$ is contained in~$V'$, and hence~$r$ has the maximum AV score~$d$ in~$E$. Then, by the definition of the tie-breaking order~$\rhd$, the first candidate from~$R$ in~$\rhd$ that has the maximum AV score $d$ in~$E$ is firstly added into the winning $k$-committee of seq$\varphi$ at~$E$. However, this contradicts that~$X$ is the winning $k$-committee of seq$\varphi$ at~$E$.  Therefore,~$B'$ dominates~$R$, implying that the {\prob{RBDS}} instance is a {\yesins}.
\end{proof}

\subsection{Adding Candidates}
{
Now we start our exploration on destructive control by adding candidates. First, we observe that {\probb{DCAC}{AV}} can be solved by utilizing the polynomial-time algorithm for the utility-based {\probb{DCAC}{AV}} presented in~\cite{DBLP:journals/jair/MeirPRZ08}: assign to each distinguished candidate utility zero, assign to each other candidate, be registered or unregistered, utility one, and set the threshold to be~$k$, the size of the winning committee. As our {\probb{DCAC}{AV}} problem is a special case of the utility-based version in~\cite{DBLP:journals/jair/MeirPRZ08}, the algorithm for {\probb{DCAC}{AV}} is actually quite trivial: Given an instance $((C\muplus D), k, J, \ell)$, let~$s$ be the maximum AV score of distinguished candidates, let $D'\subseteq D$ be the set of candidates having AV score at least $s+1$, and let $\ell'=\min\{\ell, \abs{D'}\}$. Then, the instance is a {\yesins} if and only if there are at least~$k-\ell'$ candidates in~$C\setminus J$ whose AV scores are at least~$s+1$.  The above algorithm can be also easily adapted to solve {\probb{DCAC}{seqAV}}.

However, for other concrete rules studied in the paper, we have intractability results. We first derive hardness for the two additive rules SAV and NSAV. We remark that the algorithm for {\probb{DCAC}{AV}} does not apply to SAV and NSAV because adding candidates may change the SAV and NSAV scores of candidates.
}

\begin{theorem}
\label{thm-dcac-sav-np-hard}
{\probb{DCAC}{\memph{SAV}}} is {\nph} even if there is only one distinguished candidate.
\end{theorem}

\begin{proof}
We prove Theorem~\ref{thm-dcac-sav-np-hard} via a reduction from {\prob{RX3C}}. Let $(\xs, \xc)$ be an instance of {\prob{RX3C}}. Without loss of generality, assume that $\kappa\geq 2$. In addition, we assume that~$\kappa$ is even which does not change the {\nphns} of {\prob{RX3C}}.\footnote{If~$\kappa$ is odd, we create an equivalent instance by first adding a copy of~$\xs$ and~$\xc$ into the original instance, %renaming the elements in the copied~$\xs$ and~$\xc$,
and then doubling~$\kappa$.} We create an instance of {\probb{DCAC}{SAV}} as follows. The set of registered candidates is $C=\xs\cup \{p\}$. Let $J=\{p\}$. Regarding the unregistered candidates, for every $\xce\in \xc$, we create one candidate~$c(\xce)$. In the following analysis, for a given $\xc'\subseteq \xc$, we use~$C(\xc')$ to denote the set of candidates created for elements of~$\xc'$, i.e., $C(\xc')=\{c(\xce) \setmid \xce\in \xc'\}$. Let $D=C(\xc)$. We create the following votes:
\begin{itemize}
\item a multiset of $\frac{9\kappa^2}{2}+1$ votes each of which approves exactly the~$3\kappa$ candidates in~$\xs$; (as~$\kappa$ is even, $\frac{9\kappa^2}{2}$ is an integer)
\item for each~$\xse\in \xs$, one vote $v(\xse)=\{p\}\cup \{c(\xce) \setmid \xse\in \xce\in \xc\}$ approving~$p$ and the three candidates created for $\xce\in \xc$ so that $\xse\in \xce$.
\end{itemize}
Let~$V$ be the set of the above $\frac{9\kappa^2}{2}+1+3\kappa$ votes. Let $k=3\kappa$, and let $\ell=\kappa$. In the election $(C, V)$, the SAV score of~$p$ is~$3\kappa$, and that of every $\xse\in \xs$ is $\frac{3\kappa}{2}+\frac{1}{3\kappa}$, implying that~$p$ is contained in all winning $k$-committees of  SAV at $(C, V)$. The instance of {\probb{DCAC}{SAV}} is $((C\cup D, V), k, J, \ell)$.
We prove the correctness of the reduction as follows.

$(\Rightarrow)$ Assume that~$\xc$ contains an exact set cover~$\xc'$ of~$\xs$. Let $E=(C\cup C(\xc'), V)$. By the construction of the votes, the SAV score of~$p$ in~$E$ is~$\frac{3\kappa}{2}$, the SAV score of every $\xse\in \xs$ is $\frac{3\kappa}{2}+\frac{1}{3\kappa}$, strictly larger than that of~$p$. As $\abs{\xs}=3\kappa=k$,~$p$ cannot in any winning $k$-committee of SAV at~$E$.

$(\Leftarrow)$ Assume that there exists $\xc'\subseteq \xc$ such that $\abs{\xc'}\leq \ell=\kappa$ and~$p$ is not contained in any winning $k$-committee of SAV at $E=(C\cup C(\xc'), V)$.
We claim first that~$\abs{\xc'}=\kappa$. For the sake of contradiction, assume that $\abs{\xc'}<\kappa$. Then, the SAV score of~$p$ in~$E$ is at least $3\kappa-\frac{3\abs{\xc'}}{2}\geq \frac{3\kappa}{2}+\frac{3}{2}$, which is larger than that of every candidate from $C\cup C(\xc')$ in~$E$ (every $\xse\in \xs$ has the same SAV score as in $(C, V)$, and every~$c(\xce)$ where $\xce\in \xc'$ has SAV score at most $\frac{3}{2}$). However, this contradicts that~$p$ is not contained in any winning $k$-committee of  SAV at~$E$. By this claim, we may assume now that $\abs{\xc'}=\kappa$. We show below that~$\xc'$ is an exact set cover of~$\xs$. To this end, for every $\xse\in \xs$, let $x_{\xse}=\abs{\{\xce\in \xc' \setmid \xse\in \xce\}}$ be the number of elements in~$\xc'$ containing~$\xse$. Then, the SAV score of~$p$ in the election~$E$ is
\begin{equation}
\label{eq-p-sav-score}
3\kappa-\sum_{\xse\in \xs}\left(1-\frac{1}{x_{\xse}+1}\right)=\sum_{\xse\in \xs}\frac{1}{x_{\xse}+1}.
\end{equation}
As $\abs{\xc'}=\kappa$ and every element of~$\xc'$ is a $3$-set, it holds that $\sum_{\xse\in \xs}x_{\xse}=3\kappa$. By an elementary calculation, we see that the value of~\eqref{eq-p-sav-score} achieves its minimum value $\frac{3\kappa}{2}$ only if $x_{\xse}=1$ for all $\xse\in \xs$. Otherwise, the value of~\eqref{eq-p-sav-score} is at least $\frac{3\kappa}{2}+\frac{1}{3}$ (under the fact that each~$x_a$ is a nonnegative integer) which,  provided $\kappa\geq 2$, is larger than $\frac{3\kappa}{2}+\frac{1}{3\kappa}$, the SAV score of every candidate from~$\xs$ in~$E$, contradicting that~$p$ is not contained in any  winning $k$-committee of SAV at~$E$. It is clear that~$\xc'$ is an exact set cover of~$\xs$ if and only if $x_{\xse}=1$ for all $\xse\in \xs$. Therefore, we can conclude now that~$\xc'$ is an exact set cover of~$\xs$.
% itFor the sake of contradiction, assume that this is not the case. Then, there exists $\xse\in \xs$ which is contained in at least two distinct $\xce, \xce'\in \xc'$. By the construction of the votes, the SAV score of every $c(\xce)$ where $\xce\in \xc'$ can be at most $\frac{3}{2}$, which is clearly strictly smaller than of~$p$. This implies that~$\xs$ is the unique SAV winning $k$-committee of~$E$. However, the SAV score of~$\xse$ in~$E$ can be at most $\frac{3\kappa}{2}+\frac{1}{3\kappa}$, and the SAV score of~$p$ in~$E$ is at least $\frac{3\kappa}{2}$. The SAV score of~$p$ in~$E$ must be strictly less than~$\frac{3\kappa+1}{2}$. This is only possible when there is a exact set cover of $\xs$.
\end{proof}

We can modify the reduction in the proof of Theorem~\ref{thm-dcac-sav-np-hard} to show the hardness for NSAV, based on a relation between SAV and NSAV studied in~\citep{DBLP:conf/atal/000120}.

\begin{lemma}[\cite{DBLP:conf/atal/000120}]
\label{lem-relation-sav-nsav}
Let $E=(C, V)$ be an election where $m=\abs{C}\geq 2$ and $n=\abs{V}$. Let~$B$ be a set of at least $n\cdot m^2$ candidates so that~$C$ and~$B$ are disjoint, and let $E'=(C\cup B, V)$ (so none of~$B$ is approved by any vote in~$V$). Then, for every two candidates~$c$ and~$c'$ in~$C$, it holds that ${\memph{\textsf{sc}}}_{\memph{SAV}}(\{c\}, E)>{\memph{\textsf{sc}}}_{\memph{SAV}}(\{c'\}, E)$ if and only if ${\memph{\textsf{sc}}}_{\memph{NSAV}}(\{c\}, E')>{\memph{\textsf{sc}}}_{\memph{NSAV}}(\{c'\}, E')$.
\end{lemma}

By Lemma~\ref{lem-relation-sav-nsav}, we add in the instance of {\probb{DCAC}{SAV}} constructed in the proof of Theorem~\ref{thm-dcac-sav-np-hard} a considerably large number (e.g.,~$n\cdot m^2$) of dummy candidates in~$C$ not approved by any vote, so that the SAV ranking (ranking of candidates in a nonincreasing order of their SAV scores) of candidates are identical to the NSAV ranking of candidates in  $(C\cup D', V)$ for all $D'\subseteq D$. This way, winning $k$-committees of SAV at $(C\cup D', V)$ are identical to those of NSAV at the same election, for all $D'\subseteq D$. This ensures the correctness of the reduction for NSAV. The following result arrives.

\begin{corollary}
\label{cor-dcac-nsav-np-hard}
{\probb{DCAC}{\memph{NSAV}}} is {\nph} even if there is only one distinguished candidate.
\end{corollary}

Now, we move on to sequential Thiele's rules.

\begin{theorem}
\label{thm-dcac-thiele-rule}
For each $\omega$-Thiele's rule~$\varphi$ such that $\omega(2)<2\omega(1)$,
{\probb{DCAC}{seq$\varphi$}} is {\wah} when parameterized by $k+\ell$, the size of a winning committee plus the number of added candidates. Moreover, this holds even when there is only one distinguished candidate.
\end{theorem}

\begin{proof}
We prove Theorem~\ref{thm-dcac-thiele-rule} by a reduction from {\prob{Regular Set Packing}}. Let $(\xs, \xc, \kappa)$ be an instance of {\prob{Regular Set Packing}}, where~$\xc$ is a collection of $d$-subsets of~$\xs$ for some positive integer~$d$. Let~$\varphi$ be an $\omega$-Thiele's rule such that $\omega(2)<2\omega(1)$. Notice that $\omega(1)<2\omega(1)$ implies $\omega(1)>0$. We create an instance of {\probb{DCAC}{seq$\varphi$}} as follows. For every $\xce\in \xc$ we create one candidate~$c(\xce)$. In the following discussion, for a given $\xc'\subseteq \xc$, we use~$C(\xc')=\{c(\xce) \setmid \xce\in \xc'\}$ to denote the set of candidates corresponding to~$\xc'$. In addition, we create one candidate~$p$, and create a set~$X$ of~$\kappa$ candidates. Let $C=X \cup \{p\}$, let $D=C(\xc)$, and let $J=\{p\}$. Besides, let $k=\ell=\kappa$. We create the following votes:
\begin{itemize}
\item $d$ votes approving only the distinguished candidate~$p$; and
\item for every $\xse\in \xs$, one vote $v(\xse)=\{c(\xce) \setmid \xse\in \xce\in \xc\}$ which approves all candidates corresponding to $\xce\in \xc$ containing~$\xse$.
\end{itemize}
Let~$V$ be the set of the above created~$d+\abs{\xs}$ votes. Notice that none of~$X$ is approved by any votes.
The instance of {\probb{DCAC}{seq$\varphi$}} is $((C\cup D, V), k, J, \ell)$ which can be constructed in polynomial time.
The tie-breaking order is $\overrightarrow{C(\xc)}\rhd p\rhd \overrightarrow{X}$.
Clearly, as $\omega(1)>0$ and~$p$ is before all candidates in~$X$ in the tie-breaking order,~$p$ is in the  winning $k$-committee of~seq$\varphi$ at $(C, V)$.
In the following, we show that the instance of {\prob{Regular Set Packing}} is a {\yesins} if and only if the above constructed instance of {\probb{DCAC}{seq$\varphi$}} is a {\yesins}.

$(\Rightarrow)$ Assume that~$\xc$ contains a set packing $\xc'\subseteq \xc$ of cardinality $\kappa$. Let $E=(C\cup C(\xc'), V)$. It is fairly easy to see that~$p$ is not contained in the  winning $k$-committee of seq$\varphi$ at~$E$, and hence the instance of {\probb{DCAC}{seq$\varphi$}} is a {\yesins}. In fact, as all elements of~$\xc'$ are of cardinality~$d$ and~$\xc'$ is a set packing, we know that every $c(\xce)\in C(\xc')$ where $\xce\in \xc'$ is approved by exactly~$d$ votes in~$V$ and, more importantly, none of~$V$ approves two candidates from~$C(\xc')$. As~$p$ is after all candidates in~$C(\xc)$ in the tie-breaking order and $k=\kappa$,  the  winning $k$-committee of seq$\varphi$ at~$E$ must be~$C(\xc')$% (candidates in~$C(\xc')$ added into the winning set one-by-one according to the tie-breaking order)
.

$(\Leftarrow)$ Suppose that there exists $\xc'\subseteq \xc$ such that $\abs{\xc'}\leq \kappa$ and~$p$ is not contained in the winning $k$-committee, denoted~$\w$, of seq$\varphi$ at $E=(C\cup C(\xc'), V)$. Observe first that $\abs{C(\xc')}=\kappa$. In fact, as none of $C\setminus \{p\}$ is approved by any votes, none of them can be contained in the winning $k$-committee of seq$\varphi$ at~$E$ if~$p$ is not contained in the seq$\varphi$ winning $k$-committee. Therefore, if $\abs{C(\xc')}<\kappa$,~$p$ must be contained in the winning $k$-committee of seq$\varphi$ at~$E$, which is a contradiction.
Now, we claim that~$\xc'$ is a set packing. In fact, if this is not the case, after adding some candidates from~$C(\xc')$ into the winning committee according to the procedure of seq$\varphi$, the~$\varphi$ margin contribution of~$p$ would be larger than that of any other candidates in~$D$ that have not included into the winning committee, which leads to~$p$ being contained in the winning $k$-committee and hence a contradiction arrives.
More precisely, consider the sequence of candidates~$\w^1_E$,~$\w^2_E$, $\dots$,~$\w^{k}_E$ that are added into~$\w$ one by one according to the definition of seq$\varphi$.
If~$\xc'$ is not a set packing, there exists~$i<k$ so that for every candidate $c(\xce)\in C(\xc')\setminus \w^{\leq i}_E$ there exists at least one $c(\xce')\in \w^{\leq i}_E$ so that $\xce\cap \xce'\neq\emptyset$. Let us assume that~$i$ is the minimum integer satisfying the above condition. Let~$\xse\in \xce'\cap \xce$ be any arbitrary element in the intersection of~$\xce'$ and~$\xce$, where~$\xce$ and~$\xce'$ are as above. Then, by the construction of the votes, the vote~$v(\xse)$ approves both~$c(\xce)$ and~$c(\xce')$. More precisely, by the minimality of~$i$, it holds that~$v(a)$ approves exactly the candidate~$c(\xce')$ among all in~$\w^{\leq i}_E$. As a consequence, the~$\varphi$ margin contribution of $c(\xce)$ to $\w^{\leq i}_E$ can be at most $(d-1)\cdot \omega(1)+(\omega(2)-\omega(1))$. However, the~$\varphi$ margin contribution of~$p$ to~$\w^{\leq i}_E$ is $d \cdot \omega(1)$, which is strictly larger than that of~$c(\xce)$ given $\omega(2)<2\omega(1)$. As this holds for all $c(\xce)\in C(\xc')\setminus \w^{\leq i}_E$, we obtain that $\w^{i+1}_E=p$.  From $i+1\leq k$, it follows that $p\in \w$, meaning that the instance of the {\probb{DCAC}{seq$\varphi$}} is a {\noins}.
\end{proof}

\subsection{Deleting Candidates}

Now we study destructive control by deleting candidates.
We show that any rule satisfying the Chernoff property is immune to {\prob{DCDC}}.

\begin{theorem}
\label{thm-chernoff-rule-immue-dcdc}
Every {\abmv} rule satisfying the Chernoff property is immune to {\prob{DCDC}}.
\end{theorem}

\begin{proof}
Let~$\varphi$ be an {\abmv} rule satisfying the Chernoff property.
Let $(C, V)$ be an election and let $J\subseteq C$ be a set of distinguished candidates such that there exists a winning $k$-committee~$\w$ of~$\varphi$ at $(C, V)$ with $\w\cap J\neq \emptyset$. As~$\varphi$ satisfies the Chernoff property, for any $C'\subseteq C$ such that $J\subseteq C'$ and~$\abs{C'}\geq k$, there exists at least one winning $k$-committee~$\w'$ of~$\varphi$ at~$(C', V_{C'})$ so that $(\w\cap J)\subseteq \w'$. Therefore, it is impossible to delete candidates from {$C\setminus J$} so that none of~$J$ is contained in any winning $k$-committee of~$\varphi$ at the resulting election.
\end{proof}

Recall that except AV and SAV, none of the other concrete rules studied in the paper satisfies the Chernoff property. In fact, we show below that {\prob{DCDC}} for these rules are intractable from the parameterized complexity point of view.
It is interesting to see if there are other natural {\abmv} rules satisfying the Chernoff property.

We start with the two additive rules SAV and NSAV.
For SAV, Magiera~\citeas{Magierphd2020} showed that {\prob{DCDC}} is {\nph}. We strengthen this result by deriving a {\wbhns} reduction.

\begin{theorem}
\label{thm-dcdc-sav-np-hard}
{\probb{DCDC}{\memph{SAV}}} and {\probb{DCDC}{\memph{NSAV}}} are {\wbh} with respect to~$\ell$, the number of deleted candidates. This holds even if there is only one distinguished candidate.
%every vote approves at most four candidates, and every candidate is approved by at most three votes.
\end{theorem}

\begin{proof}
We prove the theorem via a reduction from {\prob{RBDS}}. Consider first the result for SAV. Let $(G, \kappa)$ be an instance of {\prob{RBDS}}, where~$G$ is a bipartite graph with two disjoint independent sets~$R$ and~$B$. Without loss of generality, let us assume that~$R\neq \emptyset$. Additionally, as in the proof of Theorem~\ref{thm-dcdv-seq-wbh}, we assume that all red vertices have the same degree~$d$ for some positive integer~$d$. We create an instance of {\probb{DCDC}{SAV}}  as follows.
For each $\vere\in R\cup B$, we create one candidate denoted by~$c(\vere)$. In the following discussion, for a given $B'\subseteq B$ (resp.~$R'\subseteq R$), we use~$C(B')$ (resp.~$C(R')$) to denote the set of candidates corresponding to~$B'$ (resp.~$R'$), i.e.,  $C(B')=\{c(b) \setmid b\in B'\}$ (resp.~$C(R')=\{c(r) \setmid r\in R'\}$).
Then, we create one candidates~$p$, and create a set~$X$ of~$d$ candidates.
Let $J=\{p\}$, and let $C=C(R)\cup C(B)\cup X\cup J$. We create the following votes:
\begin{itemize}
\item First, we create one vote approving exactly the $d+1$ candidates in~$\{p\}\cup X$.
\item Second, we create~$\abs{R}$ votes each approving exactly~$p$.
\item Then,  for each $r\in R$, we create~$\abs{R}$ votes each approving exactly the candidate~$c(r)$.
\item Finally, for each $r\in R$, we create one vote~$v(r)$ approving~$c(r)$ and the~$d$ candidates corresponding to its neighbors in~$G$, i.e., $v(r)=\{c(r)\}\cup C(\neighbor{G}{r})$.
\end{itemize}
Let~$V$ denote the multiset of the above created votes. Clearly, $\abs{V}=1+\abs{R} + \abs{R}^2+\abs{R}=1+2 \abs{R}  + \abs{R}^2$.
Let $k=\abs{R}$ and let $\ell=\kappa$. The instance of {\probb{DCDC}{SAV}} is $((C, V), k, J, \ell)$.
In the election~$(C,V)$, the SAV score of every candidate in $\{p\}\cup C(R)$ is $\abs{R}+\frac{1}{d+1}$, that of every candidate in~$X$ is~$\frac{1}{d+1}$, and that of every~$c(b)\in C(B)$ is~$\frac{\degree{G}{b}}{d+1}$. Note that $\degree{G}{b}\leq \abs{R}$ for all $b\in B$, and it holds that $\abs{R}>0$. Hence, the winning $k$-committees of SAV at $(C, V)$ are exactly all $k$-subsets of $\{p\}\cup C(R)$.  % no matter which at most $\kappa$ candidates from $C\setminus \{p\}$ are deleted.
We prove the correctness of the reduction as follows.

$(\Rightarrow)$ Assume that~$B$ contains a subset~$B'$ of~$\kappa=\ell$ vertices which dominate~$R$. Let $E=(C\setminus C(B'), V)$. By the construction of the votes, the SAV score of~$p$ in~$E$ remains the same as in $(C, V)$. Consider now the SAV score of candidates in~$C(R)$ in the election~$E$. As~$B'$ dominates~$R$, by the construction of the votes, for every~$r\in R$, there exists at least one~$b\in B'$ so that~$v(r)$ approves~$c(b)$. It follows that the SAV score of~$c(r)$ in~$E$ increases to at least $\abs{R}+\frac{1}{d}$, which is strictly larger than that of~$p$. As this holds for all the~$\abs{R}$ candidates in~$C(R)$ and $k=\abs{R}$, it follows that~$p$ cannot be in any winning $k$-committee of SAV at~$E$.

$(\Leftarrow)$ Assume that there exists $C'\subseteq C\setminus \{p\}$ such that $\abs{C'}\leq \ell=\kappa$ and~$p$ is not contained in any winning $k$-committee of SAV at $E=(C\setminus C', V)$. Notice that, by the construction of the votes, the SAV score of every candidate from~$C(B)\setminus C'$ in the election~$E$ can be at most~$\frac{\degree{G}{b}}{2}<\abs{R}$, and that of everyone from~$X\setminus C'$ can be at most~$\frac{1}{2}$, both strictly smaller than that of~$p$ provided that $R\neq\emptyset$. As a consequence, we know that~$R$ must be the unique winning $k$-committee of SAV at~$E$. It follows that $C'\cap  C(R)=\emptyset$. Let $B'=\{b\in B \setmid c(b)\in C'\}$.
Obviously, $\abs{B'}\leq \abs{C'}\leq \kappa$.
We claim that~$B'$ dominates~$\xs$. To prove this, assume for contradiction that there exists $r\in R$ not dominated by any vertex in~$B'$. Then, by the construction of the votes, the SAV score of~$c(r)$ in~$E$ remains $\abs{R}+\frac{1}{d+1}$, the same as in the election~$(C, V)$. However,~$p$ has at least the same SAV score in~$E$ as in $(C, V)$, and given $k=\abs{R}$, this implies that there exists a winning $k$-committee of SAV at~$E$ which contains~$p$, a contradiction. Then, from $\abs{B'}\leq \kappa$ and that~$B'$ dominates~$R$, it follows that the {\prob{RBDS}} instance is a {\yesins}.

Based on Lemma~\ref{lem-relation-sav-nsav}, we can show the same {\wbhns} result for {\probb{DCDC}{NSAV}} by adding a considerable large number of dummy candidates not approved by any vote.
\end{proof}

\onlyfull{Based on Lemma~\ref{lem-relation-sav-nsav}, we can show the same {\wbhns} result for {\probb{DCDC}{NSAV}} by adding a considerable large number of dummy candidates in~$C$. So, we have the following corollary.

\begin{corollary}
\label{cor-dcdc-nsav-np-hard}
{\probb{DCDC}{\memph{NSAV}}} is {\wbh} with respect to~$\ell$, the number of deleted candidates. This holds even if there is only one distinguished candidate.
\end{corollary}
}

Finally, we study the complexity of {\prob{DCDC}} for sequential $\omega$-Thiele's rules.

\begin{theorem}
\label{thm-dcdc-seqThiele-wah}
For every $\omega$-Thiele's rule~ $\varphi$ such that $\omega(2)<2\omega(1)$, {\probb{DCDC}{\memph{seq}$\varphi$}} is {\wah} when parameterized by $m-\ell$, the number of  candidates not deleted. Moreover, this holds even when there is only one distinguished candidate.
\end{theorem}

\begin{proof}
We prove the theorem by a reduction from {\prob{Regular Set Packing}}. Let~$\varphi$ be an $\omega$-Thiele's rule~ $\varphi$ such that $\omega(2)<2\omega(1)$. Notice that by the definition of $\omega$-Thiele's rule, $\omega(2)<2\omega(1)$ implies that $\omega(1)>0$.
Let $(\xs, \xc, \kappa)$ be an instance of {\prob{Regular Set Packing}} where every $\xce\in \xc$ is a $d$-subset of~$\xs$ for some positive integer~$d$. We construct an instance of {\probb{DCDC}{seq$\varphi$}} as follows. %For each $\xse\in \xs$, we create one candidate denoted by the same symbol for simplicity.
For each $\xce\in \xc$, we create one candidate~$c(\xce)$. In the following discussion, for a given $\xc'\subseteq \xc$, we use~$C(\xc')$ to denote the set of candidates corresponding to~$\xc'$, i.e., $C(\xc')=\{c(\xce) \setmid \xce\in \xc'\}$.
In addition, we create one candidate~$p$. Let $C=C(\xc)\cup \{p\}$ and let~$J=\{p\}$.
Clearly, $\abs{C}=\abs{\xc}+1$.
We create the following votes:
\begin{itemize}
\item $d$ votes each approving only the candidate~$p$;
\item for each $\xse\in \xs$, one vote $v(\xse)=\{c(\xce) \setmid \xse\in \xce\in \xc\}$.
\end{itemize}
Let~$V$ be the multiset of the above $d+\abs{\xs}$ created votes. Finally, let $k=\kappa$ and $\ell=\abs{\xc}-\kappa$. The instance of {\probb{DCDC}{seq$\varphi$}} is $((C, V), k, J, \ell)$ which can be constructed in polynomial time. We assume that in the tie-breaking order,~$p$ is ranked in the last position. In the following, we prove the correctness of the reduction.

$(\Rightarrow)$ Assume that~$\xc$ contains a set packing~$\xc'$ of cardinality~$\kappa$. Let $E=(C(\xc')\cup \{p\}, V)$. By the construction of the votes, in the election~$E$ it holds that each candidate $c\in C(\xc')\cup \{p\}$ is approved by exactly~$d$ votes and, moreover, none of these~$d$ votes approves any other candidate from $C(\xc')\cup \{p\}$ except~$c$. Then, by the tie-breaking order and the fact that $\omega(1)>0$, the winning $k$-committee of seq$\varphi$ at~$E$ consists of the first~$k$ candidates from~$C(\xc')$ in the tie-breaking order. So, the above constructed instance of {\probb{DCDC}{seq$\varphi$}} is a {\yesins}.

$(\Leftarrow)$ Assume that~$\xc$ does not contain any set packing of cardinality~$\kappa$. Let~$\xc'$ be a subcollection of~$\xc$ so that $\abs{\xc'}\geq \abs{C}-1-\ell=\kappa$. Let $E=(C(\xc')\cup \{p\}, V)$. We claim that~$p$ is contained in the winning $k$-committee, denoted~$\w$, of seq$\varphi$ at~$E$. To this end, let us consider the sequence~$\w^1_E$,~$\w^2_E$, $\dots$,~$\w^k_E$ of candidates added into~$\w$ according to the procedure of seq$\varphi$. As~$\xc'$ does not contain any set packing of cardinality~$\kappa=k$, there exists~$i<k$ so that for every candidate $c(\xce)\in C(\xc')\setminus \w^{\leq i}_E$ there exists at least one $c(\xce')\in \w^{\leq i}_E$ so that $\xce\cap \xce'\neq\emptyset$. Let us assume that~$i$ is the minimum integer satisfying the condition. Let~$\xse\in \xce'\cap \xce$ be any arbitrary element in the intersection of~$\xce$ and~$\xce'$, where~$\xce$ and~$\xce'$ are stipulated as above. Then, by the construction of the votes, the vote~$v(\xse)$ approves both~$c(\xce)$ and~$c(\xce')$. More precisely, by the minimality of~$i$, it holds that~$v(\xse)$ approves exactly the candidate~$c(\xce')$ among all in~$\w^{\leq i}_E$, i.e., it holds that $v(\xse)\cap \w^{\leq i}_E=\{c(\xce')\}$. As a consequence, the~$\varphi$ margin contribution of~$c(\xce)$ to~$\w^{\leq i}_E$ can be at most $(d-1)\cdot \omega(1)+(\omega(2)-\omega(1))$. However, the~$\varphi$ margin contribution of~$p$ to~$\w^{\leq i}_E$ is $d \cdot \omega(1)$, which is strictly larger than that of~$c(\xce)$ given $\omega(2)<2\omega(1)$. It follows that $\w^{i+1}_E=p$. On top of that, because $i+1\leq k$, we have that $p\in \w$. Thus, the instance of the {\probb{DCDC}{seq$\varphi$}} is a {\noins}.
\end{proof}

It should be noted that the {\wahns} stated in Theorem~\ref{thm-dcdc-seqThiele-wah} also holds with respect to the parameter~$k$, because in our reduction we had indeed $k=\kappa$. However, because $m-\ell\geq k$, the {\wahns} with respect to $m-\ell$ already implies the {\wahns} when parameterized by~$k$, and thus we did not explicitly pointed it out.

Moreover, in light of Observation~\ref{obs-a}, the reduction in the proof of Theorem~\ref{thm-dcdc-seqThiele-wah} implies the following corollary.

\begin{corollary}
\label{cor-dcdc-seqThiele-np-hard-voters-3-canddiates-3}
For every sequential $\omega$-Thiele's rule $\varphi$ such that $\omega(2)<2\omega(1)$, {\probb{DCDC}{seq$\varphi$}} is {\nph} even when there is only one distinguished candidate, every voter approves at most three candidates, and every candidate is approved by at most three votes.
\end{corollary}

\section{Concluding Remarks}
In this section, we first provide a short summary of our main contribution, then we make some remarks on the complexity of control problems with respect to the axiomatic properties studied in the paper, and finally we lay out some directions for future research.

\subsection{Summary of Our Main Contribution}
We continued the line of research on the (parameterized) complexity of election control problems for several prominent {\abmv} rules including AV, SAV, NSAV, PAV, ABCCV, MAV, seqPAV, and seqABCCV. We first studied constructive control for some sequential $\omega$-Thiele's rules including seqPAV and seqABCCV, which significantly complements the work of Yang~\citeas{DBLP:conf/ijcai/Yang19} who studied constructive control for other rules. Then, we investigated destructive control for the aforementioned rules. Our investigation leads to a complete picture of the complexity landscape of the eight standard control problems for the above mentioned rules. We also established a lot of parameterized complexity results with respect to many parameters.
We emphasize that, unlike most of the previous works, to obtain the above results, we explored numerous axiomatic properties of these rules (e.g., Chernoff property, $\alpha$-efficiency, etc.) and derive general results applying to rules satisfying certain properties. We believe that these properties might be of independent interest and deserve further investigation. Our main results are summarized in Table~\ref{tab-results-summary}.
%However, it is important to point out that as discussed earlier Table~\ref{tab-results-summary} does not exhaustively summarize all our ,
%We have studied the complexity of election control for many {\abmv} rules including  AV, SAV, NSAV, PAV, ABCCV, MAV, several sequential $\omega$-Thiele's rules, etc., and obtained a comprehensive understanding of the (parameterized)  complexity these problems. Our main results are summarized in Table~\ref{tab-results-summary}.
%However, it should be pointed out that the table is not exhaustive, because
%,
%many of our results are general in the sense that they apply to a class of rules satisfying certain (axiomatic) properties.
%For instance, for sequential $\omega$-Thiele's rules, we used the following conditions: (1) $\triangle_{\omega}(i)<\omega(1)$, $i\geq 1$; (2) $\triangle_{\omega}(i)<\omega(1)$, $i\in \{1, 2, 3\}$; (3) $\triangle_{\omega}(1)<\omega(1)$; (4) $\triangle_{\omega}(i)>0$, $i\geq 0$; and (5) $\omega(1)>0$. There implication relation is $1\Rightarrow 2 \Rightarrow 3\Rightarrow 5 \Leftarrow 4$.

\subsection{Remarks on the Complexity w.r.t.\ some Axiomatic Properties}
We would  like to remark that our results established in terms of axiomatic properties are tight in the sense that if a complexity result of a destructive control problem for a rule satisfying a particular property is not shown in the paper, it means that there exists at least one rule which satisfies the property but disobeys the result. We also make some notes on the complexity of constructive control with respect to these properties. We say that a rule resists or is resistance to a control action if the corresponding control problem for the rule is computationally hard ({\nph}, {\wah}, {\wbh}, etc.); and say that the rule is vulnerable to the corresponding control action otherwise.

\paragraph{$\alpha$-Efficient and Neutral} We showed that {\prob{DCAV}} for all $\alpha$-efficient and neutral rules are computationally hard (Theorem~\ref{thm-dcav-many-rules-wbh}).
We show that there exists an $\alpha$-efficient and neutral rule (despite not a natural rule) which is vulnerable to all the other three destructive control problems. Given an election $(C, V)$ and a positive integer~$k\leq \abs{C}$, this rule first computes the set
$C_0=\{c\in C \setmid V(c)=\emptyset\}$ of candidates not approved by any votes.
Then, if $\abs{C\setminus C_0}<k$, all $k$-committees of~$C$ are winning $k$-committees. Otherwise, all $k$-committees of $C\setminus C_0$ are winning $k$-committees. This rule is clearly $\alpha$-efficient and neutral. It is easy to see that for this rule, {\prob{DCDV}} is polynomial time solvable:
\begin{itemize}
    \item Let $I=((C, V), k, J, \ell)$ be an instance of {\prob{DCDV}}, where $J\subseteq C$ is the nonempty set of distinguished candidates.
    \item If none of the distinguished candidates in~$J$ is contained in any winning $k$-committees, we  conclude that~$I$ is a {\yesins}.
    \item Otherwise, if $\abs{C\setminus C_0}<k$, we conclude that~$I$ is a {\noins}. %of if and only if $k\geq \abs{J}$ (by the definition of the rule, in this case there exists a distinguished candidate which is contained in at least one winning $k$-committee if and only if $k\geq \abs{J}$).
    \item Otherwise, we delete all votes from~$V$ approving at least one distinguished candidates. If there are more than~$\ell$ votes deleted, we conclude that~$I$ is a {\noins}.  Otherwise, we compute the set $C'$ of candidates not approved by any of the remaining votes. Then, we conclude that~$I$ is a {\yesins} if and only if $\abs{C\setminus C'}\geq k$.
\end{itemize}

{\prob{DCAC}} for this rule is also polynomial-time solvable:
\begin{itemize}
\item Let $I=((C\cup D, V), k, J, \ell)$ be an instance of {\prob{DCAC}}.
    \item If none of the distinguished candidates is contained in any winning $k$-committee of~$(C, V)$, we conclude that~$I$ is a {\yesins}.
    \item Otherwise,~$I$ is a {\noins} (in this case, $\abs{C\setminus C_0}\geq k$, and hence by the definition of the rule, we cannot turn any candidate from being contained in at least one winning $k$-committee to being not contained in any winning $k$-committee by including candidates).
\end{itemize}

Furthermore, it is easy to verify that the above rule is immune to {\prob{DCDC}}, {\prob{CCAV}}, and {\prob{CCAC}}, and that {\prob{CCDV}} and {\prob{CCDC}} for this rule are polynomial-time solvable.

\paragraph{AV-Uniform} We showed that {\prob{DCAV}} and {\prob{DCDV}} for all AV-uniform rules are computationally hard (Theorems~\ref{thm-dcav-many-rules-wbh-k-ell}--\ref{thm-dcdv-av-wah-k-ell}), {\prob{DCAC}} for AV is polynomial-time solvable, and AV is immune to {\prob{DCDC}}.
The hardness of {\prob{CCAV}} and {\prob{CCDV}} for $r$-approval~\cite{DBLP:journals/corr/abs-1005-4159} imply that {\prob{CCAV}} and {\prob{CCDV}} for all AV-uniform rules are computationally hard. Additionally, it is known that AV is immune to {\prob{CCAC}}, and {\prob{CCDC}} for AV is polynomial-time solvable.
%We show below that there is an AV-uniform rule for which both {\prob{CCAC}} and {\prob{CCDC}} are polynomial-time solvable. In particular, this rule determines the winning $k$-committees as follows: if all votes approve the same number of candidates, winning $k$-committees are computed with respect to the AV rule; otherwise, all $k$-committees are winning $k$-committees.  This rule is clearly AV-uniform, and the polynomial-time solvability of {\prob{CCAC}} and {\prob{CCDC}} for this rule is easy to verify.

\begin{table*}[ht]
\caption{Complexity of election control with respect to several axiomatic properties of {\abmv} rules. Here, an entry filled with ``R'' means that the corresponding problem is computationally hard for all rules satisfying the corresponding properties. Besides, an entry with ``\sout{R}'' (resp.\ \sout{I}) means that there is at least one rule which satisfies the corresponding property but is vulnerable (resp.\ susceptible) to the corresponding control problem.}
\label{tab-complexity-results-properties-tight}
    \centering
    \begin{tabular}{|l|c|c|c|c|c|c|c|c|}\hline
         &  CCAV & CCDV & CCAC & CCDC & DCAV & DCDV & DCAC & DCDC \\ \hline

        $\alpha$-efficient \& neutral & \sout{R} &\sout{R}&\sout{R}&\sout{R}&R&\sout{R}&\sout{R}&\sout{R} \\ \hline

        AV-uniform &R&R&\sout{R}&\sout{R}&R&R&\sout{R}&\sout{R}\\ \hline

        Chernoff &\sout{I}&\sout{I}&\sout{I}&\sout{I}&\sout{I}&\sout{I}&\sout{I}&I \\ \hline
    \end{tabular}
\end{table*}

\paragraph{Chernoff Property} We showed that all Chernoff rules are immune to {\prob{DCDC}} (Theorem~\ref{thm-chernoff-rule-immue-dcdc}). We show below that there exists Chernoff rules which are not immune to any of the other seven types of control actions. First, AV satisfies Chernoff property but is susceptible to {\prob{DCAV}},  {\prob{DCDV}}, and {\prob{DCAC}}, {\prob{CCAV}}, {\prob{CCDV}}, and {\prob{CCDC}}. AV is immune to {\prob{CCAC}, but we derive a rule satisfying the Chernoff property but is susceptible to {\prob{CCAC}}. Consider a rule which determines the single winners (winning $1$-committees) of elections with three candidates~$a$,~$b$, and~$c$ as follows. For an election with only the two candidates~$a$ and~$b$,~$b$ is the winner, and for an election with all three candidates~$a$,~$b$, and~$c$,~$a$ is the winner. Let $C=\{a, b\}$, $D=\{c\}$, and $J=\{a\}$. Clearly, we can add~$c$ to~$C$ so that~$a$ becomes the winner.

We refer to Table~\ref{tab-complexity-results-properties-tight} for a summary of the above discussion.

We also remark  that in our results for sequential $\omega$-Thiele's rules, the following conditions are used:
\begin{center}
\begin{minipage}{0.67\textwidth}
\begin{enumerate}
    \item[(1)] $\triangle_{\omega}(i)<\omega(1)$, $i\geq 1$  $\cdots \cdots \cdots \cdots \cdots \cdots \cdots \cdots \cdots \cdot$ \hfill Theorems~\ref{ccav-seqpav-nph},~\ref{thm-ccac-seqThiele-wbh}
    \item[(2)] $\triangle_{\omega}(i)<\omega(1)$, $i\in [3]$ $\cdots \cdots \cdots \cdots\cdots\cdots\cdots \cdots\cdots\cdots \cdot$ \hfill Theorem~\ref{thm-ccdc-seqpav-seqabccv-np-hard}
    \item[(3)] $\triangle_{\omega}(1)<\omega(1)$ (i.e., $\omega(2)<2\omega(1)$) $\cdots \cdots \cdot \cdot$  \hfill Theorems~\ref{ccdv-seqpav-nph},~\ref{thm-dcac-thiele-rule},~\ref{thm-dcdc-seqThiele-wah}
    \item[(4)] $\triangle_{\omega}(i)>0$, $i\geq 0$ $\cdots \cdots \cdots \cdots \cdots \cdots \cdots \cdots \cdots \cdots \cdots \cdot \cdots $  \hfill Theorem~\ref{thm-dcav-seqpav-nph}
    \item[(5)] $\omega(1)>0$ \& $\triangle_{\omega}(1)\leq \omega(1)$ $\cdots \cdots \cdots \cdots \cdots \cdots \cdots$  \hfill Theorems~\ref{thm-dcav-seqabccv-wa-hard},~\ref{thm-dcdv-seq-wbh}
\end{enumerate}
\end{minipage}
\end{center}
Their implication relation is $(1)\Rightarrow (2) \Rightarrow (3) \Rightarrow (5) \Leftarrow (4)$, where an arrow from~$(x)$ to~$(y)$ means that~$(x)$ implies~$(y)$. If there is no arc between two conditions, it means that the two conditions are independent.
These conditions are quite natural in the sense that most of the important sequential $\omega$-Thiele's rules satisfy them. For instance, arguably the  most important sequential rules seqPAV and seqABCCV satisfy even Condition~(1), the strongest among all except Condition~(4). As a matter of fact, there are other rules like sequential Webster approval voting\footnote{This is the $\omega$-Thiele's rule such that $\omega(i)=\sum_{j=1}^i\frac{1}{2j-1}$ for all $i\geq 1$.}  which satisfies  Condition~(1) too. seqABCCV violates Condition~(4), but Theorem~\ref{thm-dcav-seqpav-nph} is the only one using this condition, and we have Theorem~\ref{thm-dcav-seqabccv-wa-hard} as a supplement  (which applies to all the above sequential rules\onlyfull{, though the {\wahns} reduction is targeted for a different parameter}). seqAV satisfies only Condition~(5), and hence by Theorems~\ref{thm-dcav-seqabccv-wa-hard} and~\ref{thm-dcdv-seq-wbh}, {\probb{DCAV}{seqAV}} and {\probb{DCDV}{seqAV}} are respectively {\wah} and {\wbh} with respect to the corresponding parameters. Other results for seqAV can be inherited  from (or trivially implied by) those for AV studied in the literature. In summary, we have the complexity of all the eight control problems for all of the above mentioned sequential rules.

\subsection{Future Research}
There are many interesting questions left for future research. For instance, do the hardness results of control by adding/deleting candidates become fixed-parameter tractable with respect to the number of votes? Will the complexity of these problems change radically (from {\nph} to {\poly}) when restricted to some specific domains of dichotomous preferences (e.g., the domains of candidates interval, voters interval, etc.)?  We refer to~\citep{DBLP:conf/ijcai/ElkindL15,DBLP:journals/corr/abs-2205-09092,DBLP:journals/aarc/Karpov22,DBLP:conf/ijcai/Yang19a} for the definitions of many restricted domains widely studied in the literature.

%Notice that as {\prob{Set Packing}} remains {\nph} even when each $\xce\in \xc$ is of cardinality three and every~$\xse$ is contained in three elements of~$\xc$~\cite{}, the reduction in the proof of Theorem~\ref{thm-dcdc-seqThiele-wah} implies the following corollary.

%\begin{corollary}
%\label{cor-dcdc-seqThiele-np-hard-voters-3-canddiates-3}
%For every sequential $\omega$-Thiele's rule $\varphi$ such that $\omega(2)<2\omega(1)$, {\probb{DCDC}{seq$\varphi$}} is {\nph} even when there is only one distinguished candidate, every voter approves at most three %candidates, and every candidate is approved by at most three votes.
%\end{corollary}

\section{Acknowledgments}
The author thanks all anonymous reviewers who have provided instructive comments on the paper.

\end{document}